\documentclass[twoside]{aiml16}
\usepackage{aiml16macro,amsfonts,amsmath,amssymb,color,graphicx,latexsym}
\usepackage{yfonts,bm,enumitem,multicol}
\def\sf{\mathtt{sf}}
\def\N{\mathbb{N}}
\def\degre{\mathtt{deg}}
\def\NP{\mathbf{NP}}
\def\PSPACE{\mathbf{PSPACE}}
\def\L{\mathbf{L}}
\def\S{\mathbf{S}}
\def\K{\mathbf{K}}
\def\KD{\mathbf{KD}}
\def\KB{\mathbf{KB}}
\def\KT{\mathbf{KT}}
\def\KTB{\mathbf{KTB}}
\def\Alt{\mathbf{Alt}}
\def\LTL{\mathbf{LTL}}
\def\T{\mathbf{T}}
\def\D{\mathbf{D}}
\def\VAR{\mathbf{VAR}}
\def\BIT{\mathbf{BIT}}
\def\FOR{\mathbf{FOR}}
\def\SUB{\mathbf{SUB}}
\def\CHA{\mathbf{CHA}}
\def\for{\mathbf{for}}
\def\lastname{Balbiani \& Gencer \& Rostamigiv \& Tinchev}
\begin{document}
\begin{frontmatter}
\title{About the unification types of the modal logics
determined by classes of deterministic frames}
\author{Philippe Balbiani$^{a}$
\hspace{0.2cm}
\c{C}i\u{g}dem Gencer$^{a,b}$}
\author{Maryam Rostamigiv$^{a}$
\hspace{0.2cm}
Tinko Tinchev$^{c}$}
\address{$^{a}$Toulouse Institute of Computer Science Research
\\
CNRS~---~Toulouse University, Toulouse, France
\\
$^{b}$Faculty of Arts and Sciences
\\
Istanbul Ayd\i n University, Istanbul, Turkey
\\
$^{c}$Faculty of Mathematics and Informatics
\\
Sofia University St. Kliment Ohridski, Sofia, Bulgaria}
\begin{abstract}
The unification problem in a propositional logic is to determine, given a formula $\varphi$, whether there exists a substitution $\sigma$ such that $\sigma(\varphi)$ is in that logic.
In that case, $\sigma$ is a unifier of $\varphi$.
When a unifiable formula has minimal complete sets of unifiers, the formula is either infinitary, finitary, or unitary, depending on the cardinality of its minimal complete sets of unifiers.
In this paper, we prove that for all $d{\geq}2$, in modal logic $\Alt_{1}+\square^{d}\bot$, unifiable formulas are unitary.
\end{abstract}
\end{frontmatter}
\section{Introduction}\label{section:introduction}
The unification problem in a propositional logic is to determine, given a formula $\varphi$, whether there exists a substitution $\sigma$ such that $\sigma(\varphi)$ is in that logic.
In that case, $\sigma$ is a unifier of $\varphi$.
We shall say that a set of unifiers of a unifiable formula $\varphi$ is complete if for all unifiers $\sigma$ of $\varphi$, there exists a unifier $\tau$ of $\varphi$ in that set such that $\tau$ is more general than $\sigma$.
Now, an important question is to determine whether a given unifiable formula has minimal complete sets of unifiers~\cite{Baader:Ghilardi:2011}.
When such sets exist, they all have the same cardinality.
In that case, a unifiable formula is either infinitary, or finitary, or unitary, depending whether its complete sets of unifiers are either infinite, or finite, or with cardinality $1$.
Otherwise, the formula is nullary.
\\
\\
Within the context of the unification problem in a propositional logic, we usually distinguish between elementary unification and unification with constants.
In unification with constants, some variables (called constants) are never replaced by formulas when one applies a substitution whereas in elementary unification, all variables are likely to be replaced.
About the unification type of modal logics\footnote{In this paper, all modal logics are normal.
We follow the same conventions as in~\cite{Blackburn:deRijke:Venema:2001,Chagrov:Zakharyaschev:1997,Kracht:1999} for talking about them: $\S5$ is the least modal logic containing the formulas usually denoted $\T$, $4$ and $5$, $\KD$ is the least modal logic containing the formula usually denoted $\D$, etc.
In other respect, $\LTL$ is the modal logic with ``next'' and ``until'' interpreted over $\N$.
For more on $\LTL$, see~\cite{Emerson:1990,Goldblatt:1992}.}, it is known that $\KT$, $\KD$ and $\KB$ are nullary~\cite{Balbiani:2019,Balbiani:Gencer:2017a,Balbiani:Gencer:to:appear}, $\S5$ and $\S4.3$ are unitary~\cite{Dzik:2003,Dzik:2007,Dzik:Wojtylak:2012}, transitive modal logics like $\K4$ and $\S4$ are finitary~\cite{Ghilardi:2000,Iemhoff:2016}, $\KD45$ and $\K45$ are unitary~\cite{Ghilardi:Sacchetti:2004,Jerabek:2013}, $\K$ is nullary~\cite{Jerabek:2015} and $\K4\D1$ is unitary~\cite{Kost:2018}, the nullariness of $\KT$, $\KD$ and $\KB$ having only been obtained within the context of unification with constants.
\\
\\
The importance of the unification problem lies in its connection with the admissibility problem.
In a consistent propositional logic $\L$, unification is reducible to non-admissibility, seeing that the unifiability in $\L$ of a formula $\varphi$ is equivalent to the non-admissibility in $\L$ of the inference rule $\frac{\varphi}{\bot}$.
As observed by Ghilardi~\cite{Ghilardi:2000}, when $\L$ has a decidable membership problem and $\L$ is either unitary, or finitary, algorithms for computing minimal complete sets of unifiers in $\L$ can be used as a key component of algorithms for solving the admissibility problem in $\L$, seeing that the admissibility in $\L$ of an inference rule $\frac{\varphi_{1},\ldots,\varphi_{p}}{\psi}$ is equivalent to the inclusion in $\L$ of the set $\{\sigma(\psi):\ \sigma{\in}\Sigma\}$, where $\Sigma$ is an arbitrary minimal complete set of unifiers of $\varphi_{1}\wedge\ldots\wedge\varphi_{p}$ in $\L$.
\\
\\
$\LTL$ is the standard modal logic used in the specification and verification of reactive systems~\cite{Emerson:1990,Goldblatt:1992}.
Owing to its significance within the context of applied non-classical logics, it is natural to answer the question of its unification type.
This has been done by Babenyshev and Rybakov~\cite{Babenyshev:Rybakov:2010} who have proved that $\LTL$ is unitary within the context of elementary unification.
It is also natural to answer the question of the unification type of $\LTL$ when its syntax is restricted somehow or other.
For instance, one can consider the syntactic restriction of $\LTL$ to its next fragment.
When interpreted over $\N$, this syntactic restriction is equivalent to $\Alt_{1}+\lozenge\top$ (the least modal logic containing all formulas of the form $\lozenge\varphi\rightarrow\square\varphi$ and the formula $\lozenge\top$).
\\
\\
There is no link between the unification type of a propositional logic and the unification types of its syntactic restrictions\footnote{For instance, Boolean Logic is unitary~\cite{Martin:Nipkow:1989} while its implication fragment is finitary within the context of unification with constants~\cite{Balbiani:Mojtahedi:submitted}.
Of course, seeing that the unification type of an equational theory depends not only on the equational theory itself but also on the set of function symbols that can occur in the considered unification problems, this phenomenon is already well-known from the theory of unification~\cite{Baader:Snyder:2001}.}.
About the unification type of $\Alt_{1}$ (the least modal logic containing all formulas of the form $\lozenge\varphi\rightarrow\square\varphi$) and its extensions, the line of reasoning determining in~\cite{Balbiani:Gencer:2017a} the unification type (nullary) of $\KD$ within the context of unification with constants can be adapted to $\Alt_{1}+\lozenge\top$ whereas the line of reasoning determining in~\cite{Jerabek:2015} the unification type (nullary) of $\K$ has been adapted to $\Alt_{1}$~\cite{Balbiani:Tinchev:2016}.
In this paper, within the context of elementary unification, we prove that for all $d{\geq}2$, in $\Alt_{1}+\square^{d}\bot$ (the least modal logic containing all formulas of the form $\lozenge\varphi\rightarrow\square\varphi$ and the formula $\square^{d}\bot$), unifiable formulas are unitary\footnote{We assume the reader is at home with tools and techniques in modal logics.
For more on them, see~\cite{Blackburn:deRijke:Venema:2001,Chagrov:Zakharyaschev:1997,Kracht:1999}.}.
\section{Preliminaries}
In this section, we introduce a handful of definitions that will be useful throughout the paper.
We also introduce a result (Proposition~\ref{main:one}) that will be useful in Section~\ref{section:about:the:proof:of:the:existence:of:f:bounded:deterministic} immediately after the proof of Proposition~\ref{cardinality:ordered:appropriate:way}.
For all sets $S$, ${\parallel}S{\parallel}$ will denote the cardinality of $S$.
For all nonempty sets $S$, for all equivalence relations $\sim$ on $S$ and for all $\alpha{\in}S$, $\lbrack\alpha\rbrack$ will denote the equivalence class modulo $\sim$ with $\alpha$ as its representative.
For all nonempty sets $S$, for all equivalence relations $\sim$ on $S$ and for all $T{\subseteq}S$, $T/{\sim}$ will denote the quotient set of $T$ modulo $\sim$.
Notice that for all nonempty sets $S$, for all equivalence relations $\sim$ on $S$ and for all $\alpha,\beta{\in}S$, $\alpha{\sim}\beta$ iff $\alpha{\in}\lbrack\beta\rbrack$ iff $\lbrack\alpha\rbrack\cap\lbrack\beta\rbrack{\not=}\emptyset$.
\begin{proposition}\label{main:one}
Let $S, T$ be finite nonempty sets.
Let $\sim$ be an equivalence relation on $S$.
The following conditions are equivalent:
\begin{enumerate}
\item ${\parallel}S/{\sim}{\parallel}{\leq}{\parallel}T{\parallel}{\leq}{\parallel}S{\parallel}$,
\item there exists a surjective function $f$ from $S$ to $T$ such that for all $\alpha,\beta{\in}S$, if $f(\alpha){=}f(\beta)$ then $\alpha{\sim}\beta$.
\end{enumerate}
\end{proposition}
Let $\VAR$ be a countably infinite set of {\em variables}\/ (with typical members denoted $x$, $y$, etc).
Let $(x_{1},x_{2},\ldots)$ be an enumeration of $\VAR$ without repetitions.
A {\em frame}\/ is a couple $(W,R)$ where $W$ is a non-empty set (with typical members denoted $s$, $t$, etc) and $R$ is a binary relation on $W$.
We shall say that a frame $(W,R)$ is {\em deterministic}\/ if for all $s,t,u{\in}W$, if $sRt$ and $sRu$ then $t{=}u$.
For all $d{\geq}2$, we shall say that a frame $(W,R)$ is {\em $d$-bounded}\/ if for all $s_{0},\ldots,s_{d}{\in}W$, there exists $i{\in}\N$ such that $i{<}d$ and not $s_{i}Rs_{i+1}$.
For all $d{\geq}2$, let ${\mathcal C}^{d}_{det}$ be the class of all deterministic $d$-bounded frames.
For all $n{\geq}1$, an {\em $n$-tuple of bits}\/ (denoted $\alpha$, $\beta$, etc) is a function from $\{1,\ldots,n\}$ to $\{0,1\}$.
For all $n{\geq}1$, let $\BIT_{n}$ be the set of all $n$-tuples of bits.
The well-founded strict partial order $\ll$ on $\N\times\N$ is defined by
\begin{itemize}
\item $(d^{\prime},d^{\prime\prime}){\ll}(d^{\prime\prime\prime},d^{\prime\prime\prime\prime})$ iff $d^{\prime}{<}d^{\prime\prime\prime}$ and $d^{\prime\prime}{<}d^{\prime\prime\prime\prime}$,
\end{itemize}
where $(d^{\prime},d^{\prime\prime}),(d^{\prime\prime\prime},d^{\prime\prime\prime\prime})$ range over $\N\times\N$.
\section{Syntax}\label{section:syntax}
Let $n{\geq}1$.
The set $\FOR_{n}$ of all {\em $n$-formulas}\/ (with typical members denoted $\varphi$, $\psi$, etc) is inductively defined as follows:
\begin{itemize}
\item $\varphi,\psi::=x_{i}\mid\bot\mid\neg\varphi\mid(\varphi\vee\psi)\mid\square\varphi$,
\end{itemize}
where $i$ ranges over $\{1,\ldots,n\}$.
We adopt the standard rules for omission of the parentheses.
The Boolean connectives $\top$, $\wedge$, $\rightarrow$ and $\leftrightarrow$ are defined by the usual abbreviations.
The modal connective $\lozenge$ is defined by
%
%
%
%
$\lozenge\varphi::=\neg\square\neg\varphi$.
%
%
%
%
For all $\varphi{\in}\FOR_{n}$, we write ``$\varphi^{0}$'' to mean ``$\neg\varphi$'' and we write ``$\varphi^{1}$'' to mean ``$\varphi$''.
For all $d{\in}\N$, the modal connective $\square^{d}$ is inductively defined as follows:
\begin{itemize}
\item $\square^{0}\varphi::=\varphi$,
\item $\square^{d+1}\varphi::=\square\square^{d}\varphi$.
\end{itemize}
The {\em $n$-degree}\/ of $\varphi{\in}\FOR_{n}$ (in symbols $\degre_{n}(\varphi)$) is the nonnegative integer inductively defined as follows:
\begin{itemize}
\item $\degre_{n}(x_{i}){=}0$,
\item $\degre_{n}(\bot){=}0$,
\item $\degre_{n}(\neg\varphi){=}\degre_{n}(\varphi)$,
\item $\degre_{n}(\varphi\vee\psi){=}\max\{\degre_{n}(\varphi),\degre_{n}(\psi)\}$,
\item $\degre_{n}(\square\varphi){=}\degre_{n}(\varphi)+1$.
\end{itemize}
\section{Semantics}\label{section:semantics}
Let $n{\geq}1$.
An {\em $n$-model}\/ based on a frame $(W,R)$ is a triple $(W,R,V)$ where $V$ is a function assigning for all $i{\in}\{1,\ldots,n\}$, a subset $V(x_{i})$ of $W$ to the variable $x_{i}$.
Given an $n$-model $(W,R,V)$, the {\em $n$-satisfiability}\/ of $\varphi{\in}\FOR_{n}$ at $s{\in}W$ (in symbols $s\models_{n}\varphi$) is inductively defined as usual.
In particular:
\begin{itemize}
%
%
%
%
%
%
%
%
%
%
\item $s{\models_{n}}\square\varphi$ iff for all $t{\in}W$, if $sRt$ then $t{\models_{n}}\varphi$.
\end{itemize}
Obviously,
%
%
%
%
$s{\models_{n}}\lozenge\varphi$ iff there exists $t{\in}W$ such that $sRt$ and $t{\models_{n}}\varphi$.
%
%
%
%
We shall say that $\varphi{\in}\FOR_{n}$ is {\em $n$-true}\/ in an $n$-model $(W,R,V)$ if $\varphi$ is $n$-satisfied at all $s{\in}W$.
We shall say that $\varphi{\in}\FOR_{n}$ is {\em $n$-valid}\/ in a frame $(W,R)$ if $\varphi$ is $n$-true in all $n$-models based on $(W,R)$.
We shall say that $\varphi{\in}\FOR_{n}$ is {\em $n$-valid}\/ in a class ${\mathcal C}$ of frames (in symbols ${\mathcal C}{\models}\varphi$) if $\varphi$ is $n$-valid in all frames in ${\mathcal C}$.
Let ${\mathcal C}$ be a class of frames.
Let $\equiv_{{\mathcal C}}^{n}$ be the equivalence relation on $\FOR_{n}$ defined by
\begin{itemize}
\item $\varphi{\equiv_{{\mathcal C}}^{n}}\psi$ iff ${\mathcal C}{\models}\varphi\leftrightarrow\psi$,
\end{itemize}
where $\varphi,\psi$ range over $\FOR_{n}$.
We shall say that ${\mathcal C}$ is {\em locally $n$-tabular}\/ if $\equiv_{{\mathcal C}}^{n}$ possesses finitely many equivalence classes.
The next result follows from~\cite[Proposition~$2.29$]{Blackburn:deRijke:Venema:2001} and the fact that for all $d{\geq}2$ and for all $\varphi{\in}\FOR_{n}$, there exists $\psi{\in}\FOR_{n}$ such that $\degre_{n}(\psi){<}d$ and ${\mathcal C}^{d}_{det}{\models}\varphi\leftrightarrow\psi$.
\begin{proposition}\label{proposition:locally:tabular:classes:of:frames}
For all $d{\geq}2$, ${\mathcal C}^{d}_{det}$ is locally $n$-tabular.
\end{proposition}
For all $d{\geq}2$, let ${\mathcal{GC}}_{det}^{d}$ be the class consisting of all frames of the form $(W,R)$ where $W{=}\{s:\ s{\in}\N$ and $0{\leq}s{\leq}d^{\prime}\}$ and $R{=}\{(s,t):\ s,t{\in}W$ and $t{-}s{=}1\}$ for some $d^{\prime}{\in}\N$ such that $d^{\prime}{<}d$.
Notice that for all $d{\geq}2$, ${\mathcal{GC}}_{det}^{d}{\subseteq}{\mathcal C}^{d}_{det}$.
The next result shows that for all $d{\geq}2$, ${\mathcal C}^{d}_{det}$ and ${\mathcal{GC}}_{det}^{d}$ determine the same modal logic: $\Alt_{1}+\square^{d}\bot$.
Its proof is standard.
\begin{proposition}\label{proposition:C:determines:alt1}
Let $d{\geq}2$.
For all $\varphi{\in}\FOR_{n}$, $\varphi{\in}\Alt_{1}+\square^{d}\bot$ iff ${\mathcal C}^{d}_{det}{\models}\varphi$ iff ${\mathcal{GC}}_{det}^{d}{\models}\varphi$.
\end{proposition}
\section{Unification}
Let $n{\geq}1$.
An {\em $n$-substitution}\/ is a couple $(k,\sigma)$ where $k{\geq}1$ and $\sigma$ is a homomorphism from $\FOR_{n}$ to $\FOR_{k}$, i.e. $\sigma$ is a function from $\FOR_{n}$ to $\FOR_{k}$ such that
\begin{itemize}
\item $\sigma(\bot){=}\bot$,
\item $\sigma(\neg\varphi){=}\neg\sigma(\varphi)$,
\item $\sigma(\varphi\vee\psi){=}\sigma(\varphi)\vee\sigma(\psi)$,
\item $\sigma(\square\varphi){=}\square\sigma(\varphi)$,
\end{itemize}
where $\varphi,\psi$ range over $\FOR_{n}$.
Let $\SUB_{n}$ be the set of all $n$-substitutions.
Let ${\mathcal C}$ be a class of frames.
The equivalence relation $\simeq_{{\mathcal C}}^{n}$ on $\SUB_{n}$ is defined by
\begin{itemize}
\item $(k,\sigma){\simeq_{{\mathcal C}}^{n}}(l,\tau)$ iff for all $i{\in}\{1,\ldots,n\}$, ${\mathcal C}{\models}\sigma(x_{i})\leftrightarrow\tau(x_{i})$,
\end{itemize}
where $(k,\sigma),(l,\tau)$ range over $\SUB_{n}$.
The preorder $\preccurlyeq_{{\mathcal C}}^{n}$ on $\SUB_{n}$ is defined by
\begin{itemize}
\item $(k,\sigma){\preccurlyeq_{{\mathcal C}}^{n}}(l,\tau)$ iff there exists a $k$-substitution $(m,\upsilon)$ such that for all $i{\in}\{1,\ldots,n\}$, ${\mathcal C}{\models}\upsilon(\sigma(x_{i}))\leftrightarrow\tau(x_{i})$,
\end{itemize}
where $(k,\sigma),(l,\tau)$ range over $\SUB_{n}$.
Obviously, $\simeq_{{\mathcal C}}^{n}$ is contained in $\preccurlyeq_{{\mathcal C}}^{n}$.
A {\em $({\mathcal C},n)$-unifier}\/ of $\varphi{\in}\FOR_{n}$ is an $n$-substitution $(k,\sigma)$ such that ${\mathcal C}{\models}\sigma(\varphi)$.
We shall say that $\varphi{\in}\FOR_{n}$ is {\em $({\mathcal C},n)$-unifiable}\/ if there exists a $({\mathcal C},n)$-unifier of $\varphi$.
We shall say that a set $\Sigma$ of $({\mathcal C},n)$-unifiers of a $({\mathcal C},n)$-unifiable $\varphi{\in}\FOR_{n}$ is {\em $({\mathcal C},n)$-complete}\/ if for all $({\mathcal C},n)$-unifiers $(k,\sigma)$ of $\varphi$, there exists $(l,\tau){\in}\Sigma$ such that $(l,\tau){\preccurlyeq_{{\mathcal C}}^{n}}(k,\sigma)$.
The next result is standard.
\begin{proposition}\label{mcs:have:the:same:cardinalities}
Let $\varphi{\in}\FOR_{n}$.
If $\varphi$ is $({\mathcal C},n)$-unifiable then for all minimal $({\mathcal C},n)$-complete sets $\Sigma,\Delta$ of $({\mathcal C},n)$-unifiers of $\varphi$, ${\parallel}\Sigma{\parallel}{=}{\parallel}\Delta{\parallel}$.
\end{proposition}
An important question is the following: when $\varphi{\in}\FOR_{n}$ is $({\mathcal C},n)$-unifiable, is there a minimal $({\mathcal C},n)$-complete set of $({\mathcal C},n)$-unifiers of $\varphi$?
When the answer is ``yes'', how large is this set?
For all $({\mathcal C},n)$-unifiable $\varphi{\in}\FOR_{n}$, we shall say that: $\varphi$ is {\em $({\mathcal C},n)$-nullary}\/ if there exists no minimal $({\mathcal C},n)$-complete set of $({\mathcal C},n)$-unifiers of $\varphi$; $\varphi$ is {\em $({\mathcal C},n)$-infinitary}\/ if there exists a minimal $({\mathcal C},n)$-complete set of $({\mathcal C},n)$-unifiers of $\varphi$ but there exists no finite one; $\varphi$ is {\em $({\mathcal C},n)$-finitary}\/ if there exists a finite minimal $({\mathcal C},n)$-complete set of $({\mathcal C},n)$-unifiers of $\varphi$ but there exists no with cardinality $1$; $\varphi$ is {\em $({\mathcal C},n)$-unitary}\/ if there exists a minimal $({\mathcal C},n)$-complete set of $({\mathcal C},n)$-unifiers of $\varphi$ with cardinality $1$.
%
%
%
%
%
%
%
%
%
%
%
%
%
%
Obviously, the types ``nullary'', ``infinitary'', ``finitary'' and ``unitary'' constitute a set of jointly exhaustive and pairwise distinct situations for each unifiable $n$-formula.
We shall say that: ${\mathcal C}$ is {\em $n$-nullary}\/ if there exists a $({\mathcal C},n)$-nullary $({\mathcal C},n)$-unifiable $n$-formula; ${\mathcal C}$ is {\em $n$-infinitary}\/ if every $({\mathcal C},n)$-unifiable $n$-formula possesses a minimal $({\mathcal C},n)$-complete set of $({\mathcal C},n)$-unifiers and there exists a $({\mathcal C},n)$-infinitary $({\mathcal C},n)$-unifiable $n$-formula; ${\mathcal C}$ is {\em $n$-finitary}\/ if every $({\mathcal C},n)$-unifiable $n$-formula possesses a finite minimal $({\mathcal C},n)$-complete set of $({\mathcal C},n)$-unifiers and there exists a $({\mathcal C},n)$-finitary $({\mathcal C},n)$-unifiable $n$-formula; ${\mathcal C}$ is {\em $n$-unitary}\/ if every $({\mathcal C},n)$-unifiable $n$-formula possesses a minimal $({\mathcal C},n)$-complete set of $({\mathcal C},n)$-unifiers with cardinality $1$.
%
%
%
%
%
%
%
%
%
%
%
%
%
%
Obviously, the types ``nullary'', ``infinitary'', ``finitary'' and ``unitary'' constitute a set of jointly exhaustive and pairwise distinct situations for each class of frames.
For all $({\mathcal C},n)$-unifiable $\varphi{\in}\FOR_{n}$, we shall say that
%
%
%
%
$\varphi$ is {\em $({\mathcal C},n)$-filtering}\/ if for all $({\mathcal C},n)$-unifiers $(k,\sigma),(l,\tau)$ of $\varphi$, there exists a $({\mathcal C},n)$-unifier $(m,\upsilon)$ of $\varphi$ such that $(m,\upsilon){\preccurlyeq_{{\mathcal C}}^{n}}(k,\sigma)$ and $(m,\upsilon){\preccurlyeq_{{\mathcal C}}^{n}}(l,\tau)$.
%
%
%
%
See~\cite{Ghilardi:Sacchetti:2004,Jerabek:2013} for further discussion about filtering unification.
The next result is standard.
\begin{proposition}\label{lemma:filtering:implies:unitary:or:nullary}
Let $\varphi{\in}\FOR_{n}$ be $({\mathcal C},n)$-unifiable.
If $\varphi$ is $({\mathcal C},n)$-filtering then either $\varphi$ is $({\mathcal C},n)$-nullary, or $\varphi$ is $({\mathcal C},n)$-unitary.
\end{proposition}
For all $({\mathcal C},n)$-unifiable $\varphi{\in}\FOR_{n}$ and for all $\pi{\geq}1$, we shall say that
%
%
%
%
$\varphi$ is {\em $({\mathcal C},n)$-$\pi$-reasonable} if for all $({\mathcal C},n)$-unifiers $(k,\sigma)$ of $\varphi$, if $k{>}\pi$ then there exists a $({\mathcal C},n)$-unifier $(l,\tau)$ of $\varphi$ such that $(l,\tau){\preccurlyeq_{{\mathcal C}}^{n}}(k,\sigma)$ and $l{\leq}\pi$.
%
%
%
%
The next result is new.
\begin{proposition}\label{lemma:locally:tabular:implies:unitary:or:finitary}
Let $\varphi{\in}\FOR_{n}$ be $({\mathcal C},n)$-unifiable and $\pi{\geq}1$.
If ${\mathcal C}$ is locally $\pi$-tabular and $\varphi$ is $({\mathcal C},n)$-$\pi$-reasonable then either $\varphi$ is $({\mathcal C},n)$-finitary, or $\varphi$ is $({\mathcal C},n)$-unitary.
\end{proposition}
\begin{proof}
Suppose ${\mathcal C}$ is locally $\pi$-tabular and $\varphi$ is $({\mathcal C},n)$-$\pi$-reasonable.
Let $\Sigma$ be the set of all $({\mathcal C},n)$-unifiers of $\varphi$.
Notice that $\Sigma$ is $({\mathcal C},n)$-complete.
Let $\Sigma^{\prime}$ be the set of $n$-substitutions obtained from $\Sigma$ by keeping only the $n$-substitutions $(k,\sigma)$ such that $k{\leq}\pi$.
Since $\Sigma$ is $({\mathcal C},n)$-complete and $\varphi$ is $({\mathcal C},n)$-$\pi$-reasonable, therefore $\Sigma^{\prime}$ is $({\mathcal C},n)$-complete.
Let $\Sigma^{\prime\prime}$ be the set of $n$-substitutions obtained from $\Sigma^{\prime}$ by keeping only one representative of each equivalence class modulo $\simeq_{{\mathcal C}}^{n}$.
Since $\Sigma^{\prime}$ is $({\mathcal C},n)$-complete, therefore $\Sigma^{\prime\prime}$ is $({\mathcal C},n)$-complete.
Moreover, since ${\mathcal C}$ is locally $\pi$-tabular, therefore $\Sigma^{\prime\prime}$ is finite.
Hence, either $\varphi$ is $({\mathcal C},n)$-finitary, or $\varphi$ is $({\mathcal C},n)$-unitary.
\end{proof}
\section{About bounded deterministic frames}\label{section:unification:type:alt1:boxdbot}
Let $n{\geq}1$.
Let $d{\geq}2$.
Combined with Proposition~\ref{lemma:filtering:implies:unitary:or:nullary}, the next result implies that in $\Alt_{1}+\square^{d}\bot$, unifiable $n$-formulas are either nullary, or unitary.
\begin{proposition}\label{ALT1:bounded:is:filtering}
For all $\varphi{\in}\FOR_{n}$, if $\varphi$ is $({\mathcal {GC}}^{d}_{det},n)$-unifiable then $\varphi$ is $({\mathcal {GC}}^{d}_{det},n)$-filtering.
\end{proposition}
\begin{proof}
Let $\varphi{\in}\FOR_{n}$.
Suppose $\varphi$ is $({\mathcal {GC}}^{d}_{det},n)$-unifiable.
Let $(k,\sigma),(l,\tau)$ be $({\mathcal {GC}}^{d}_{det},n)$-unifiers of $\varphi$.
Let $m{=}max\{k,l\}{+}1$.
Let $(m,\mu)$ be the $n$-substitution defined by
\begin{itemize}
\item $\mu(x_{i}){=}(\bigvee\{\lozenge^{l}(x_{m}\wedge\square\bot):\ 0{\leq}l{<}d\}\wedge\sigma(x_{i}))\vee(\bigwedge\{\square^{l}(\neg x_{m}\vee\lozenge\top):\ 0{\leq}l{<}d\}\wedge\tau(x_{i}))$,
\end{itemize}
where $i$ ranges over $\{1,\ldots,n\}$.
Let $(m,\lambda_{\top})$ and $(m,\lambda_{\bot})$ be the $m$-substitutions defined by
\begin{itemize}
\item if $i{<}m$ then $\lambda_{\top}(x_{i}){=}x_{i}$ else $\lambda_{\top}(x_{i}){=}\top$,
\item if $i{<}m$ then $\lambda_{\bot}(x_{i}){=}x_{i}$ else $\lambda_{\bot}(x_{i}){=}\bot$,
\end{itemize}
where $i$ ranges over $\{1,\ldots,m\}$.
Notice that for all $i{\in}\{1,\ldots,n\}$, ${\mathcal {GC}}^{d}_{det}\models\lambda_{\top}(\mu(x_{i}))\leftrightarrow\sigma(x_{i})$ and ${\mathcal {GC}}^{d}_{det}\models\lambda_{\bot}(\mu(x_{i}))\leftrightarrow\tau(x_{i})$.
Hence, $(m,\mu){\preccurlyeq_{{\mathcal {GC}}^{d}_{det}}^{n}}(k,\sigma)$ and $(m,\mu){\preccurlyeq_{{\mathcal {GC}}^{d}_{det}}^{n}}(l,\tau)$.
Moreover, by induction on $\psi{\in}\FOR_{n}$ the reader may show that ${\mathcal {GC}}^{d}_{det}\models\bigvee\{\lozenge^{l}(x_{m}\wedge\square\bot):\ 0{\leq}l{<}d\}\rightarrow (\mu(\psi)\leftrightarrow\sigma(\psi))$ and ${\mathcal {GC}}^{d}_{det}\models\bigwedge\{\square^{l}(\neg x_{m}\vee\lozenge\top):\ 0{\leq}l{<}d\}\rightarrow (\mu(\psi)\leftrightarrow\tau(\psi))$.
Thus, ${\mathcal {GC}}^{d}_{det}\models\bigvee\{\lozenge^{l}(x_{m}\wedge\square\bot):\ 0{\leq}l{<}d\}\rightarrow\mu(\varphi)$ and ${\mathcal {GC}}^{d}_{det}\models\bigwedge\{\square^{l}(\neg x_{m}\vee\lozenge\top):\ 0{\leq}l{<}d\}\rightarrow\mu(\varphi)$.
Consequently, ${\mathcal {GC}}^{d}_{det}\models\mu(\varphi)$ and $(m,\mu)$ is a $({\mathcal {GC}}^{d}_{det},n)$-unifier of $\varphi$.
Since $(m ,\mu){\preccurlyeq_{{\mathcal {GC}}^{d}_{det}}^{n}}(k,\sigma)$ and $(m ,\mu){\preccurlyeq_{{\mathcal {GC}}^{d}_{det}}^{n}}(l,\tau)$, therefore $\varphi$ is $({\mathcal {GC}}^{d}_{det},n)$-filtering.
\end{proof}
In order to show that in $\Alt_{1}+\square^{d}\bot$, unifiable $n$-formulas are reasonable (Proposition~\ref{main:proposition:bounded:deterministic}), we introduce an alternative semantics based on chains.
For all $d^{\prime}{\in}\N$, if $d^{\prime}{<}d$ then a {\em $d^{\prime}$-$n$-chain}\/ is a structure of the form $(\alpha^{0},\ldots,\alpha^{d^{\prime}})$ where $\alpha^{0},\ldots,\alpha^{d^{\prime}}{\in}\BIT_{n}$.
For all $d^{\prime}{\in}\N$, if $d^{\prime}{<}d$ then let $\CHA_{{=}d^{\prime}}^{n}$ be the set of all $d^{\prime}$-$n$-chains.
Let $\CHA_{d}^{n}{=}\bigcup\{\CHA_{{=}d^{\prime}}^{n}:\ d^{\prime}{\in}\N$ and $d^{\prime}{<}d\}$.
The binary relation $\models_{n}$ between $\CHA_{d}^{n}$ and $\FOR_{n}$ is defined by
\begin{itemize}
\item $(\alpha^{0},\ldots,\alpha^{d^{\prime}}){\models_{n}}x_{i}$ iff $\alpha^{0}_{i}{=}1$,
\item $(\alpha^{0},\ldots,\alpha^{d^{\prime}}){\not\models_{n}}\bot$,
\item $(\alpha^{0},\ldots,\alpha^{d^{\prime}}){\models_{n}}\neg\varphi$ iff $(\alpha^{0},\ldots,\alpha^{d^{\prime}}){\not\models_{n}}\varphi$,
\item $(\alpha^{0},\ldots,\alpha^{d^{\prime}}){\models_{n}}\varphi\vee\psi$ iff either $(\alpha^{0},\ldots,\alpha^{d^{\prime}}){\models_{n}}\varphi$, or $(\alpha^{0},\ldots,\alpha^{d^{\prime}}){\models_{n}}\psi$,
\item $(\alpha^{0},\ldots,\alpha^{d^{\prime}}){\models_{n}}\square\varphi$ iff if $d^{\prime}{\geq}1$ then $(\alpha^{1},\ldots,\alpha^{d^{\prime}}){\models_{n}}\varphi$.
\end{itemize}
Obviously,
\begin{itemize}
\item $(\alpha^{0},\ldots,\alpha^{d^{\prime}}){\models_{n}}\lozenge\varphi$ iff $d^{\prime}{\geq}1$ and $(\alpha^{1},\ldots,\alpha^{d^{\prime}}){\models_{n}}\varphi$.
\end{itemize}
The next result shows that ${\mathcal{GC}}_{det}^{d}$ and chains determine the same modal logic.
Its proof is standard.
\begin{proposition}\label{proposition:C:determines:alt1:encore}
For all $\varphi{\in}\FOR_{n}$, ${\mathcal{GC}}_{det}^{d}{\models}\varphi$ iff for all $(\alpha^{0},\ldots,\alpha^{d^{\prime}}){\in}\CHA_{d}^{n}$, $(\alpha^{0},\ldots,\alpha^{d^{\prime}}){\models_{n}}\varphi$.
\end{proposition}
The function $\for_{n}$ from $\CHA_{d}^{n}$ to $\FOR_{n}$ is inductively defined as follows:
\begin{itemize}
\item if $d^{\prime}{\geq}1$ then $\for_{n}((\alpha^{0},\ldots,\alpha^{d^{\prime}})){=}x_{1}^{\alpha^{0}_{1}}\wedge\ldots\wedge x_{n}^{\alpha^{0}_{n}}\wedge\lozenge\for_{n}((\alpha^{1},\ldots,\alpha^{d^{\prime}}))$ else $\for_{n}((\alpha^{0},\ldots,\alpha^{d^{\prime}})){=}x_{1}^{\alpha^{0}_{1}}\wedge\ldots\wedge x_{n}^{\alpha^{0}_{n}}\wedge\square\bot$,
\end{itemize}
where $(\alpha^{0},\ldots,\alpha^{d^{\prime}})$ ranges over $\CHA_{d}^{n}$.
In Propositions~\ref{simple:proposition:about:formulas:associated:to:chains}--\ref{lemma:3:about:k:alpha:A:gamma:C:B:equals:C}, we study its main properties.
\begin{proposition}\label{simple:proposition:about:formulas:associated:to:chains}
Let $(k,\sigma){\in}\SUB_{n}$.
Let $(\alpha^{0},\ldots,\alpha^{d^{\prime}}){\in}\CHA_{d}^{k}$ and $(\beta^{0},\ldots,\beta^{d^{\prime\prime}}){\in}\CHA_{d}^{n}$.
If $(\alpha^{0},\ldots,\alpha^{d^{\prime}}){\models_{k}}\sigma(\for_{n}((\beta^{0},\ldots,\beta^{d^{\prime\prime}})))$ then $d^{\prime}{=}d^{\prime\prime}$.
\end{proposition}
\begin{proof}
By $\ll$-induction on $(d^{\prime},d^{\prime\prime})$.
\end{proof}
\begin{proposition}\label{lemma:1:about:alpha:beta:equivalent:conditions}
Let $(\alpha^{0},\ldots,\alpha^{d^{\prime}}),(\beta^{0},\ldots,\beta^{d^{\prime\prime}}){\in}\CHA_{d}^{n}$.
The following conditions are equivalent:
\begin{enumerate}
\item $(\alpha^{0},\ldots,\alpha^{d^{\prime}}){=}(\beta^{0},\ldots,\beta^{d^{\prime\prime}})$,
\item $(\alpha^{0},\ldots,\alpha^{d^{\prime}}){\models_{n}}\for_{n}((\beta^{0},\ldots,\beta^{d^{\prime\prime}}))$.
\end{enumerate}
\end{proposition}
\begin{proof}
By $\ll$-induction on $(d^{\prime},d^{\prime\prime})$.
\end{proof}
\begin{proposition}\label{lemma:2:about:k:alpha:sigma:beta:B}
Let $(k,\sigma){\in}\SUB_{n}$.
Let $(\alpha^{0},\ldots,\alpha^{d^{\prime}}){\in}\CHA_{d}^{k}$.
There exists $(\beta^{0},\ldots,\beta^{d^{\prime\prime}}){\in}\CHA_{d}^{n}$ such that $(\alpha^{0},\ldots,\alpha^{d^{\prime}}){\models_{k}}\sigma(\for_{n}((\beta^{0},\ldots,\beta^{d^{\prime\prime}})))$.
\end{proposition}
\begin{proof}
By induction on $d^{\prime}$.
We consider the following $2$ cases.
\\
\\
{\bf Case $d^{\prime}{=}0$.}
Let $\beta^{0}{\in}\BIT_{n}$ be such that for all $i{\in}\{1,\ldots,n\}$, if $(\alpha^{0},\ldots,\alpha^{d^{\prime}}){\models_{k}}\sigma(x_{i})$ then $\beta^{0}_{i}{=}1$ else $\beta^{0}_{i}{=}0$.
Consequently, $(\alpha^{0},\ldots,\alpha^{d^{\prime}}){\models_{k}}\sigma(x_{1})^{\beta^{0}_{1}}\wedge\ldots\wedge\sigma(x_{n})^{\beta^{0}_{n}}$.
Since $d^{\prime}{=}0$, therefore $(\alpha^{0},\ldots,\alpha^{d^{\prime}}){\models_{k}}\sigma(\for_{n}((\beta^{0})))$.
\\
\\
{\bf Case $d^{\prime}{\geq}1$.}
Let $\beta^{0}{\in}\BIT_{n}$ be such that for all $i{\in}\{1,\ldots,n\}$, if $(\alpha^{0},\ldots,\alpha^{d^{\prime}}){\models_{k}}\sigma(x_{i})$ then $\beta^{0}_{i}{=}1$ else $\beta^{0}_{i}{=}0$.
Moreover, since $d^{\prime}{\geq}1$, therefore by induction hypothesis, let $(\beta^{1},\ldots,\beta^{d^{\prime\prime}}){\in}\CHA_{d}^{n}$ be such that $(\alpha^{1},\ldots,\alpha^{d^{\prime}}){\models_{k}}\sigma(\for_{n}((\beta^{1},\ldots,\beta^{d^{\prime\prime}})))$.
Hence, $(\alpha^{0},\ldots,\alpha^{d^{\prime}}){\models_{k}}\sigma(x_{1})^{\beta^{0}_{1}}\wedge\ldots\wedge\sigma(x_{n})^{\beta^{0}_{n}}$.
Moreover, $(\alpha^{0},\ldots,\alpha^{d^{\prime}}){\models_{k}}\lozenge\sigma(\for_{n}((\beta^{1},\ldots,\beta^{d^{\prime\prime}})))$.
Thus, $(\alpha^{0},\ldots,\alpha^{d^{\prime}}){\models_{k}}\sigma(\for_{n}((\beta^{0},\ldots,\beta^{d^{\prime\prime}})))$.
\end{proof}
\begin{proposition}\label{lemma:3:about:k:alpha:A:gamma:C:B:equals:C}
Let $(k,\sigma){\in}\SUB_{n}$.
Let $(\alpha^{0},\ldots,\alpha^{d^{\prime}}){\in}\CHA_{d}^{k}$.
For all $(\beta^{0},\ldots,\beta^{d^{\prime\prime}}),(\gamma^{0},\ldots,\gamma^{d^{\prime\prime\prime}}){\in}\CHA_{d}^{n}$, if $(\alpha^{0},\ldots,\alpha^{d^{\prime}}){\models_{k}}\sigma(\for_{n}((\beta^{0},\ldots,\beta^{d^{\prime\prime}})))$ and $(\alpha^{0},\ldots,\alpha^{d^{\prime}}){\models_{k}}\sigma(\for_{n}((\gamma^{0},\ldots,\gamma^{d^{\prime\prime\prime}})))$ then $(\beta^{0},\ldots,\beta^{d^{\prime\prime}}){=}(\gamma^{0},\ldots,\gamma^{d^{\prime\prime\prime}})$.
\end{proposition}
\begin{proof}
By induction on $d^{\prime}$.
We consider the following $2$ cases.
\\
\\
{\bf Case $d^{\prime}{=}0$.}
Let $(\beta^{0},\ldots,\beta^{d^{\prime\prime}}),(\gamma^{0},\ldots,\gamma^{d^{\prime\prime\prime}}){\in}\CHA_{d}^{n}$ be such that $(\alpha^{0},\ldots,\alpha^{d^{\prime}}){\models_{k}}\sigma(\for_{n}((\beta^{0},\ldots,\beta^{d^{\prime\prime}})))$ and $(\alpha^{0},\ldots,\alpha^{d^{\prime}}){\models_{k}}\sigma(\for_{n}((\gamma^{0},\ldots,\gamma^{d^{\prime\prime\prime}})))$.
Hence, if $d^{\prime\prime}{\geq}1$ then $(\alpha^{0},\ldots,\alpha^{d^{\prime}}){\models_{k}}\sigma(x_{1})^{\beta^{0}_{1}}\wedge\ldots\wedge\sigma(x_{n})^{\beta^{0}_{n}}\wedge\lozenge\sigma(\for_{n}((\beta^{1},\ldots,\beta^{d^{\prime\prime}})))$ else $(\alpha^{0},\ldots,\alpha^{d^{\prime}}){\models_{k}}\sigma(x_{1})^{\beta^{0}_{1}}\wedge\ldots\wedge\sigma(x_{n})^{\beta^{0}_{n}}\wedge\square\bot$ and if $d^{\prime\prime\prime}{\geq}1$ then $(\alpha^{0},\ldots,\alpha^{d^{\prime}}){\models_{k}}\sigma(x_{1})^{\gamma^{0}_{1}}\wedge\ldots\wedge\sigma(x_{n})^{\gamma^{0}_{n}}\wedge\lozenge\sigma(\for_{n}((\gamma^{1},\ldots,\gamma^{d^{\prime\prime\prime}})))$ else $(\alpha^{0},\ldots,\alpha^{d^{\prime}}){\models_{k}}\sigma(x_{1})^{\gamma^{0}_{1}}\wedge\ldots\wedge\sigma(x_{n})^{\gamma^{0}_{n}}\wedge\square\bot$.
Since $d^{\prime}{=}0$, therefore $d^{\prime\prime}{=}0$, $d^{\prime\prime\prime}{=}0$ and for all $i{\in}\{1,\ldots,n\}$, $(\alpha^{0},\ldots,\alpha^{d^{\prime}}){\models_{k}}\sigma(x_{i})^{\beta^{0}_{i}}$ and $(\alpha^{0},\ldots,\alpha^{d^{\prime}}){\models_{k}}\sigma(x_{i})^{\gamma^{0}_{i}}$.
Thus, $\beta^{0}{=}\gamma^{0}$.
Since $d^{\prime\prime}{=}0$ and $d^{\prime\prime\prime}{=}0$, therefore $(\beta^{0},\ldots,\beta^{d^{\prime\prime}}){=}(\gamma^{0},\ldots,\gamma^{d^{\prime\prime\prime}})$.
\\
\\
{\bf Case $d^{\prime}{\geq}1$.}
Let $(\beta^{0},\ldots,\beta^{d^{\prime\prime}}),(\gamma^{0},\ldots,\gamma^{d^{\prime\prime\prime}}){\in}\CHA_{d}^{n}$ be such that $(\alpha^{0},\ldots,\alpha^{d^{\prime}}){\models_{k}}\sigma(\for_{n}((\beta^{0},\ldots,\beta^{d^{\prime\prime}})))$ and $(\alpha^{0},\ldots,\alpha^{d^{\prime}}){\models_{k}}\sigma(\for_{n}((\gamma^{0},\ldots,\gamma^{d^{\prime\prime\prime}})))$.
Hence, if $d^{\prime\prime}{\geq}1$ then $(\alpha^{0},\ldots,\alpha^{d^{\prime}}){\models_{k}}\sigma(x_{1})^{\beta^{0}_{1}}\wedge\ldots\wedge\sigma(x_{n})^{\beta^{0}_{n}}\wedge\lozenge\sigma(\for_{n}((\beta^{1},\ldots,\beta^{d^{\prime\prime}})))$ else $(\alpha^{0},\ldots,\alpha^{d^{\prime}}){\models_{k}}\sigma(x_{1})^{\beta^{0}_{1}}\wedge\ldots\wedge\sigma(x_{n})^{\beta^{0}_{n}}\wedge\square\bot$ and if $d^{\prime\prime\prime}{\geq}1$ then $(\alpha^{0},\ldots,\alpha^{d^{\prime}}){\models_{k}}\sigma(x_{1})^{\gamma^{0}_{1}}\wedge\ldots\wedge\sigma(x_{n})^{\gamma^{0}_{n}}\wedge\lozenge\sigma(\for_{n}((\gamma^{1},\ldots,\gamma^{d^{\prime\prime\prime}})))$ else $(\alpha^{0},\ldots,\alpha^{d^{\prime}}){\models_{k}}\sigma(x_{1})^{\gamma^{0}_{1}}\wedge\ldots\wedge\sigma(x_{n})^{\gamma^{0}_{n}}\wedge\square\bot$.
Since $d^{\prime}{\geq}1$, therefore $d^{\prime\prime}{\geq}1$, $d^{\prime\prime\prime}{\geq}1$ and for all $i{\in}\{1,\ldots,n\}$, $(\alpha^{0},\ldots,\alpha^{d^{\prime}}){\models_{k}}\sigma(x_{i})^{\beta^{0}_{i}}$ and $(\alpha^{0},\ldots,\alpha^{d^{\prime}}){\models_{k}}\sigma(x_{i})^{\gamma^{0}_{i}}$.
Moreover, $(\alpha^{1},\ldots,\alpha^{d^{\prime}}){\models_{k}}\sigma(\for_{n}((\beta^{1},\ldots,\beta^{d^{\prime\prime}})))$ and $(\alpha^{1},\ldots,\alpha^{d^{\prime}}){\models_{k}}\sigma(\for_{n}((\gamma^{1},\ldots,\gamma^{d^{\prime\prime\prime}})))$.
Thus, $\beta^{0}{=}\gamma^{0}$.
Moreover, by induction hypothesis, $(\beta^{1},\ldots,\beta^{d^{\prime\prime}}){=}(\gamma^{1},\ldots,\gamma^{d^{\prime\prime\prime}})$.
Consequently, $(\beta^{0},\ldots,\beta^{d^{\prime\prime}}){=}(\gamma^{0},\ldots,\gamma^{d^{\prime\prime\prime}})$.
\end{proof}
For all $k{\geq}1$, a {\em $d$-$(k,n)$-morphism}\/ is a function $f$ from $\CHA_{d}^{k}$ to $\CHA_{d}^{n}$ such that for all $(\alpha^{0},\ldots,\alpha^{d^{\prime}}){\in}\CHA_{d}^{k}$ and for all $(\beta^{0},\ldots,\beta^{d^{\prime\prime}}){\in}\CHA_{d}^{n}$, if $f((\alpha^{0},\ldots,\alpha^{d^{\prime}})){=}(\beta^{0},\ldots,\beta^{d^{\prime\prime}})$ then
\begin{description}
\item[forward condition:] if $d^{\prime}{\geq}1$ then $d^{\prime\prime}{\geq}1$ and $f((\alpha^{1},\ldots,\alpha^{d^{\prime}})){=}(\beta^{1},\ldots,\beta^{d^{\prime\prime}})$,
\item[backward condition:] if $d^{\prime\prime}{\geq}1$ then $d^{\prime}{\geq}1$ and $f((\alpha^{1},\ldots,\alpha^{d^{\prime}})){=}(\beta^{1},\ldots,\beta^{d^{\prime\prime}})$.
\end{description}
The next result is a good example of what the properties of morphisms are like.
\begin{proposition}\label{simple:proposition:about:morphisms}
Let $k{\geq}1$.
Let $f$ be a $d$-$(k,n)$-morphism.
Let $(\beta^{0},\ldots,\beta^{d^{\prime\prime}}){\in}\CHA_{d}^{k}$ and $(\gamma^{0},\ldots,\gamma^{d^{\prime\prime\prime}}){\in}\CHA_{d}^{n}$.
If the following conditions hold then $f((\beta^{0},\ldots,\beta^{d^{\prime\prime}})){=}(\gamma^{0},\ldots,\gamma^{d^{\prime\prime\prime}})$:
\begin{itemize}
\item for all $i{\in}\{1,\ldots,n\}$, $f((\beta^{0},\ldots,\beta^{d^{\prime\prime}})){\models_{n}}x_{i}^{\gamma^{0}_{i}}$,
\item if $d^{\prime\prime}{\geq}1$ then $d^{\prime\prime\prime}{\geq}1$ and $f((\beta^{1},\ldots,\beta^{d^{\prime\prime}})){=}(\gamma^{1},\ldots,\gamma^{d^{\prime\prime\prime}})$,
\item if $d^{\prime\prime\prime}{\geq}1$ then $d^{\prime\prime}{\geq}1$ and $f((\beta^{1},\ldots,\beta^{d^{\prime\prime}})){=}(\gamma^{1},\ldots,\gamma^{d^{\prime\prime\prime}})$.
\end{itemize}
\end{proposition}
\begin{proof}
Suppose for all $i{\in}\{1,\ldots,n\}$, $f((\beta^{0},\ldots,\beta^{d^{\prime\prime}})){\models_{n}}x_{i}^{\gamma^{0}_{i}}$.
Moreover, suppose if $d^{\prime\prime}{\geq}1$ then $d^{\prime\prime\prime}{\geq}1$ and $f((\beta^{1},\ldots,\beta^{d^{\prime\prime}})){=}(\gamma^{1},\ldots,\gamma^{d^{\prime\prime\prime}})$ and if $d^{\prime\prime\prime}{\geq}1$ then $d^{\prime\prime}{\geq}1$ and $f((\beta^{1},\ldots,\beta^{d^{\prime\prime}})){=}(\gamma^{1},\ldots,\gamma^{d^{\prime\prime\prime}})$.
For the sake of the contradiction, suppose $f((\beta^{0},\ldots,\beta^{d^{\prime\prime}})){\not=}(\gamma^{0},\ldots,\gamma^{d^{\prime\prime\prime}})$.
Let $(\delta^{0},\ldots,\delta^{d^{\prime\prime\prime\prime}}){\in}\CHA_{d}^{n}$ be such that $f((\beta^{0},\ldots,\beta^{d^{\prime\prime}})){=}(\delta^{0},\ldots,\delta^{d^{\prime\prime\prime\prime}})$.
Since for all $i{\in}\{1,\ldots,n\}$, $f((\beta^{0},\ldots,\beta^{d^{\prime\prime}})){\models_{n}}x_{i}^{\gamma^{0}_{i}}$, therefore for all $i{\in}\{1,\ldots,n\}$, $(\delta^{0},\ldots,\delta^{d^{\prime\prime\prime\prime}}){\models_{n}}x_{i}^{\gamma^{0}_{i}}$.
Since for all $i{\in}\{1,\ldots,n\}$, $(\delta^{0},\ldots,\delta^{d^{\prime\prime\prime\prime}}){\models_{n}}x_{i}^{\delta^{0}_{i}}$, therefore $\gamma^{0}{=}\delta^{0}$.
Since $f((\beta^{0},\ldots,\beta^{d^{\prime\prime}})){\not=}(\gamma^{0},\ldots,\gamma^{d^{\prime\prime\prime}})$ and $f((\beta^{0},\ldots,\beta^{d^{\prime\prime}})){=}(\delta^{0},\ldots,\delta^{d^{\prime\prime\prime\prime}})$, therefore $(\gamma^{0},\ldots,\gamma^{d^{\prime\prime\prime}}){\not=}(\delta^{0},\ldots,\delta^{d^{\prime\prime\prime\prime}})$.
Since $\gamma^{0}{=}\delta^{0}$, therefore either $d^{\prime\prime\prime}{\geq}1$, or $d^{\prime\prime\prime\prime}{\geq}1$.
We consider the following $2$ cases.
\\
\\
{\bf Case $d^{\prime\prime\prime}{\geq}1$.}
Hence, $d^{\prime\prime}{\geq}1$ and $f((\beta^{1},\ldots,\beta^{d^{\prime\prime}})){=}(\gamma^{1},\ldots,\gamma^{d^{\prime\prime\prime}})$.
Since $f$ is a $d$-$(k,n)$-morphism and $f((\beta^{0},\ldots,\beta^{d^{\prime\prime}})){=}(\delta^{0},\ldots,\delta^{d^{\prime\prime\prime\prime}})$, therefore $d^{\prime\prime\prime\prime}{\geq}1$ and $f((\beta^{1},\ldots,\beta^{d^{\prime\prime}})){=}(\delta^{1},\ldots,\delta^{d^{\prime\prime\prime\prime}})$.
Since $\gamma^{0}{=}\delta^{0}$ and $f((\beta^{1},\ldots,\beta^{d^{\prime\prime}})){=}(\gamma^{1},\ldots,\gamma^{d^{\prime\prime\prime}})$, therefore $(\gamma^{0},\ldots,\gamma^{d^{\prime\prime\prime}}){=}(\delta^{0},\ldots,\delta^{d^{\prime\prime\prime\prime}})$: a contradiction.
\\
\\
{\bf Case $d^{\prime\prime\prime\prime}{\geq}1$.}
Since $f$ is a $d$-$(k,n)$-morphism and $f((\beta^{0},\ldots,\beta^{d^{\prime\prime}})){=}(\delta^{0},\ldots,\delta^{d^{\prime\prime\prime\prime}})$, therefore $d^{\prime\prime}{\geq}1$ and $f((\beta^{1},\ldots,\beta^{d^{\prime\prime}})){=}(\delta^{1},\ldots,\delta^{d^{\prime\prime\prime\prime}})$.
Thus, $d^{\prime\prime\prime}{\geq}1$ and $f((\beta^{1},\ldots,\beta^{d^{\prime\prime}})){=}(\gamma^{1},\ldots,\gamma^{d^{\prime\prime\prime}})$.
Since $\gamma^{0}{=}\delta^{0}$ and $f((\beta^{1},\ldots,\beta^{d^{\prime\prime}})){=}(\delta^{1},\ldots,\delta^{d^{\prime\prime\prime\prime}})$, therefore $(\gamma^{0},\ldots,\gamma^{d^{\prime\prime\prime}}){=}(\delta^{0},\ldots,\delta^{d^{\prime\prime\prime\prime}})$: a contradiction.
\end{proof}
\section{Main result}\label{new:section:Main:result}
Let $n{\geq}1$.
Let $d{\geq}2$.
Let $\pi{=}n$.
Notice that $n{\leq}\pi$.
Combined with Propositions~\ref{proposition:locally:tabular:classes:of:frames}, \ref{proposition:C:determines:alt1} and~\ref{lemma:locally:tabular:implies:unitary:or:finitary}, the next result implies that in $\Alt_{1}+\square^{d}\bot$, unifiable $n$-formulas are either finitary, or unitary.
\begin{proposition}\label{main:proposition:bounded:deterministic}
For all $\varphi{\in}\FOR_{n}$, if $\varphi$ is $({\mathcal{GC}}_{det}^{d},n)$-unifiable then $\varphi$ is $({\mathcal{GC}}_{det}^{d},n)$-$\pi$-reasonable.
\end{proposition}
\begin{proof}
Let $\varphi{\in}\FOR_{n}$.
Suppose $\varphi$ is $({\mathcal{GC}}_{det}^{d},n)$-unifiable.
Let $(k,\sigma)$ be a $({\mathcal{GC}}_{det}^{d},n)$-unifier of $\varphi$ such that $k{>}\pi$.
Hence, ${\mathcal{GC}}_{det}^{d}{\models}\sigma(\varphi)$.
Moreover, since $n{\leq}\pi$, therefore $k{\geq}n$.
Let $g$ be a $d$-$(k,n)$-morphism such that for all $(\alpha^{0},\ldots,\alpha^{d^{\prime}}),(\beta^{0},\ldots,\beta^{d^{\prime\prime}}){\in}\CHA_{d}^{k}$, if $g((\alpha^{0},\ldots,\alpha^{d^{\prime}})){=}g((\beta^{0},\ldots,\beta^{d^{\prime\prime}}))$ then for all $i{\in}\{1,\ldots,n\}$, $(\alpha^{0},\ldots,\alpha^{d^{\prime}}){\models_{k}}\sigma(x_{i})$ iff $(\beta^{0},\ldots,\beta^{d^{\prime\prime}}){\models_{k}}\sigma(x_{i})$.
The proof of the existence of $g$ is presented in Section~\ref{section:about:the:proof:of:the:existence:of:g:bounded:deterministic}.
Let $f$ be a surjective $d$-$(k,n)$-morphism such that for all $(\alpha^{0},\ldots,\alpha^{d^{\prime}}),(\beta^{0},\ldots,\beta^{d^{\prime\prime}}){\in}\CHA_{d}^{k}$, if $f((\alpha^{0},\ldots,\alpha^{d^{\prime}})){=}f((\beta^{0},\ldots,\beta^{d^{\prime\prime}}))$ then $g((\alpha^{0},\ldots,\alpha^{d^{\prime}})){=}g((\beta^{0},\ldots,\beta^{d^{\prime\prime}}))$.
The proof of the existence of $f$ is presented in Section~\ref{section:about:the:proof:of:the:existence:of:f:bounded:deterministic}.
Let $(n,\tau),(k,\nu)$ be the $n$-substitutions defined by
\begin{itemize}
\item $\tau(x_{i}){=}\bigvee\{\for_{n}(f((\alpha^{0},\ldots,\alpha^{d^{\prime}}))):\ (\alpha^{0},\ldots,\alpha^{d^{\prime}}){\in}\CHA_{d}^{k}$ is such that $(\alpha^{0},\ldots,\alpha^{d^{\prime}}){\models_{k}}\sigma(x_{i})\}$,
\item $\nu(x_{i}){=}\bigvee\{\for_{k}((\alpha^{0},\ldots,\alpha^{d^{\prime}})):\ (\alpha^{0},\ldots,\alpha^{d^{\prime}}){\in}\CHA_{d}^{k}$ is such that $f((\alpha^{0},\ldots,\alpha^{d^{\prime}})){\models_{n}}x_{i}\}$,
\end{itemize}
where $i$ ranges over $\{1,\ldots,n\}$.
Now, we show that $\varphi$ is $({\mathcal{GC}}_{det}^{d},n)$-$\pi$-reasonable.
This necessitates our proving Lemmas~\ref{lemma:6:about:psi:beta:B:equivalent:condi}--\ref{lemma:9:about:beta:B:nu:tau:sigma:x:i}.
%
%
\begin{lemma}\label{lemma:6:about:psi:beta:B:equivalent:condi}
Let $\psi{\in}\FOR_{n}$.
For all $(\beta^{0},\ldots,\beta^{d^{\prime\prime}}){\in}\CHA_{d}^{n}$, the following conditions are equivalent:
\begin{enumerate}
\item there exists $(\alpha^{0},\ldots,\alpha^{d^{\prime}}){\in}\CHA_{d}^{k}$ such that $f((\alpha^{0},\ldots,\alpha^{d^{\prime}})){=}(\beta^{0},\ldots,\beta^{d^{\prime\prime}})$ and $(\alpha^{0},\ldots,\alpha^{d^{\prime}}){\models_{k}}\sigma(\psi)$,
\item for all $(\alpha^{0},\ldots,\alpha^{d^{\prime}}){\in}\CHA_{d}^{k}$, if $f((\alpha^{0},\ldots,\alpha^{d^{\prime}})){=}(\beta^{0},\ldots,\beta^{d^{\prime\prime}})$ then $(\alpha^{0},\ldots,\alpha^{d^{\prime}}){\models_{k}}\sigma(\psi)$,
\item $(\beta^{0},\ldots,\beta^{d^{\prime\prime}}){\models_{n}}\tau(\psi)$.
\end{enumerate}
\end{lemma}
\begin{lemma}\label{lemma:7:about:beta:B:i:following:equivalent:conditions}
For all $(\beta^{0},\ldots,\beta^{d^{\prime\prime}}){\in}\CHA_{d}^{k}$ and for all $i{\in}\{1,\ldots,n\}$, the following conditions are equivalent:
\begin{enumerate}
\item $(\beta^{0},\ldots,\beta^{d^{\prime\prime}}){\models_{k}}\nu(x_{i})$,
\item $f((\beta^{0},\ldots,\beta^{d^{\prime\prime}})){\models_{n}}x_{i}$.
\end{enumerate}
\end{lemma}
\begin{lemma}\label{lemma:8:about:beta:B:MOD:gamma:C:are:equi}
Let $(\beta^{0},\ldots,\beta^{d^{\prime\prime}}){\in}\CHA_{d}^{k}$ and $(\gamma^{0},\ldots,\gamma^{d^{\prime\prime\prime}}){\in}\CHA_{d}^{n}$.
The following
\linebreak
conditions are equivalent:
\begin{enumerate}
\item $f((\beta^{0},\ldots,\beta^{d^{\prime\prime}})){=}(\gamma^{0},\ldots,\gamma^{d^{\prime\prime\prime}})$,
\item $(\beta^{0},\ldots,\beta^{d^{\prime\prime}}){\models_{k}}\nu(\for_{n}((\gamma^{0},\ldots,\gamma^{d^{\prime\prime\prime}})))$.
\end{enumerate}
\end{lemma}
\begin{lemma}\label{lemma:9:about:beta:B:nu:tau:sigma:x:i}
For all $(\beta^{0},\ldots,\beta^{d^{\prime\prime}}){\in}\CHA_{d}^{k}$ and for all $i{\in}\{1,\ldots,n\}$, the following conditions are equivalent:
\begin{enumerate}
\item $(\beta^{0},\ldots,\beta^{d^{\prime\prime}}){\models_{k}}\nu(\tau(x_{i}))$,
\item $(\beta^{0},\ldots,\beta^{d^{\prime\prime}}){\models_{k}}\sigma(x_{i})$.
\end{enumerate}
\end{lemma}
Since ${\mathcal{GC}}_{det}^{d}{\models}\sigma(\varphi)$, therefore by Proposition~\ref{proposition:C:determines:alt1:encore}, for all $(\alpha^{0},\ldots,\alpha^{d^{\prime}}){\in}\CHA_{d}^{k}$, $(\alpha^{0},\ldots,\alpha^{d^{\prime}}){\models_{k}}\sigma(\varphi)$.
Thus, by Lemma~\ref{lemma:6:about:psi:beta:B:equivalent:condi}, for all $(\beta^{0},\ldots,\beta^{d^{\prime\prime}}){\in}\CHA_{d}^{n}$, $(\beta^{0},\ldots,\beta^{d^{\prime\prime}}){\models_{n}}\tau(\varphi)$.
Consequently, by Proposition~\ref{proposition:C:determines:alt1:encore}, ${\mathcal{GC}}_{det}^{d}{\models}\tau(\varphi)$.
Hence, $(n,\tau)$ is a $({\mathcal{GC}}_{det}^{d},n)$-unifier of $\varphi$.
Moreover, by Lemma~\ref{lemma:9:about:beta:B:nu:tau:sigma:x:i}, $(n,\tau){\preccurlyeq_{{\mathcal{GC}}_{det}^{d}}^{n}}(k,\sigma)$.
Since $n{\leq}\pi$, therefore $\varphi$ is $({\mathcal{GC}}_{det}^{d},n)$-$\pi$-reasonable.
\end{proof}
The next result follows from Propositions~\ref{proposition:locally:tabular:classes:of:frames}, \ref{proposition:C:determines:alt1}, \ref{lemma:filtering:implies:unitary:or:nullary}, \ref{lemma:locally:tabular:implies:unitary:or:finitary}, \ref{ALT1:bounded:is:filtering} and~\ref{main:proposition:bounded:deterministic}.
\begin{proposition}\label{final:proposition:for:this:section}
For all $\varphi{\in}\FOR_{n}$, if $\varphi$ is $({\mathcal{GC}}_{det}^{d},n)$-unifiable then $\varphi$ is $({\mathcal{GC}}_{det}^{d},n)$-unitary.
\end{proposition}
Now, our main result can be stated as follows.
\begin{proposition}\label{main:result:all:unifiable:formulas:are:unitary:proposition}
${\mathcal{C}}_{det}^{d}$ and ${\mathcal{GC}}_{det}^{d}$ are $n$-unitary.
\end{proposition}
\begin{proof}
By Propositions~\ref{proposition:C:determines:alt1} and~\ref{final:proposition:for:this:section}.
\end{proof}
\section{Definition of the function $g$ used in Section~\ref{new:section:Main:result}}\label{section:about:the:proof:of:the:existence:of:g:bounded:deterministic}
Let $n{\geq}1$.
Let $d{\geq}2$.
Let $(k,\sigma){\in}\SUB_{n}$.
{\bf Now, we define the function $g$ used in Section~\ref{new:section:Main:result}.}
Let $g$ be the function from $\CHA_{d}^{k}$ to $\CHA_{d}^{n}$ such that
\begin{itemize}
\item $g((\alpha^{0},\ldots,\alpha^{d^{\prime}}))$ is the unique $(\beta^{0},\ldots,\beta^{d^{\prime\prime}}){\in}\CHA_{d}^{n}$ such that $(\alpha^{0},\ldots,\alpha^{d^{\prime}}){\models_{k}}\sigma(\for_{n}((\beta^{0},\ldots,\beta^{d^{\prime\prime}})))$,
\end{itemize}
where $(\alpha^{0},\ldots,\alpha^{d^{\prime}})$ ranges over $\CHA_{d}^{k}$.
Notice that by Propositions~\ref{lemma:2:about:k:alpha:sigma:beta:B} and~\ref{lemma:3:about:k:alpha:A:gamma:C:B:equals:C}, $g$ is well-defined.
Moreover, by Proposition~\ref{simple:proposition:about:formulas:associated:to:chains}, for all $(\alpha^{0},\ldots,\alpha^{d^{\prime}}){\in}\CHA_{d}^{k}$, $g((\alpha^{0},\ldots,\alpha^{d^{\prime}})){\in}\CHA_{{=}d^{\prime}}^{n}$.
Propositions~\ref{first:lemma:about:g:alt1} and~\ref{second:lemma:about:g:alt1} show that $g$ possesses the properties required in Section~\ref{new:section:Main:result}.
\begin{proposition}\label{first:lemma:about:g:alt1}
$g$ is a $d$-$(k,n)$-morphism.
\end{proposition}
\begin{proof}
Let $(\alpha^{0},\ldots,\alpha^{d^{\prime}}){\in}\CHA_{d}^{k}$ and $(\beta^{0},\ldots,\beta^{d^{\prime\prime}}){\in}\CHA_{d}^{n}$ be such that $g((\alpha^{0},\ldots,\alpha^{d^{\prime}})){=}(\beta^{0},\ldots,\beta^{d^{\prime\prime}})$.
Hence, $(\alpha^{0},\ldots,\alpha^{d^{\prime}}){\models_{k}}\sigma(\for_{n}((\beta^{0},\ldots,\beta^{d^{\prime\prime}})))$.
Thus, if $d^{\prime\prime}{\geq}1$ then $(\alpha^{0},\ldots,\alpha^{d^{\prime}}){\models_{k}}\sigma(x_{1})^{\beta^{0}_{1}}\wedge\ldots\wedge\sigma(x_{n})^{\beta^{0}_{n}}\wedge\lozenge\sigma(\for_{n}((\beta^{1},\ldots,\beta^{d^{\prime\prime}})))$ else $(\alpha^{0},\ldots,\alpha^{d^{\prime}}){\models_{k}}\sigma(x_{1})^{\beta^{0}_{1}}\wedge\ldots\wedge\sigma(x_{n})^{\beta^{0}_{n}}\wedge\square\bot$.
Consequently, if $d^{\prime}{\geq}1$ then $d^{\prime\prime}{\geq}1$ and $(\alpha^{1},\ldots,\alpha^{d^{\prime}}){\models_{k}}\sigma(\for_{n}((\beta^{1},\ldots,\beta^{d^{\prime\prime}})))$, i.e. $g((\alpha^{1},\ldots,\alpha^{d^{\prime}})){=}(\beta^{1},\ldots,\beta^{d^{\prime\prime}})$.
Moreover, if $d^{\prime\prime}{\geq}1$ then $d^{\prime}{\geq}1$ and $(\alpha^{1},\ldots,\alpha^{d^{\prime}}){\models_{k}}\sigma(\for_{n}((\beta^{1},\ldots,\beta^{d^{\prime\prime}})))$, i.e. $g((\alpha^{1},\ldots,\alpha^{d^{\prime}})){=}(\beta^{1},\ldots,\beta^{d^{\prime\prime}})$.
\end{proof}
\begin{proposition}\label{second:lemma:about:g:alt1}
For all $(\alpha^{0},\ldots,\alpha^{d^{\prime}}),(\beta^{0},\ldots,\beta^{d^{\prime\prime}}){\in}\CHA_{d}^{k}$, if $g((\alpha^{0},\ldots,\alpha^{d^{\prime}})){=}g((\beta^{0},\ldots,\beta^{d^{\prime\prime}}))$ then for all $i{\in}\{1,\ldots,n\}$, $(\alpha^{0},\ldots,\alpha^{d^{\prime}}){\models_{k}}\sigma(x_{i})$ iff $(\beta^{0},\ldots,\beta^{d^{\prime\prime}}){\models_{k}}\sigma(x_{i})$.
\end{proposition}
\begin{proof}
Let $(\alpha^{0},\ldots,\alpha^{d^{\prime}}),(\beta^{0},\ldots,\beta^{d^{\prime\prime}}){\in}\CHA_{d}^{k}$.
Suppose $g((\alpha^{0},\ldots,\alpha^{d^{\prime}})){=}g((\beta^{0},\ldots,\beta^{d^{\prime\prime}}))$.
Hence, let $(\gamma^{0},\ldots,\gamma^{d^{\prime\prime\prime}}){\in}\CHA_{d}^{n}$ be such that $g((\alpha^{0},\ldots,\alpha^{d^{\prime}})){=}(\gamma^{0},\ldots,\gamma^{d^{\prime\prime\prime}})$ and $g((\beta^{0},\ldots,\beta^{d^{\prime\prime}})){=}(\gamma^{0},\ldots,\gamma^{d^{\prime\prime\prime}})$.
Thus, $(\alpha^{0},\ldots,\alpha^{d^{\prime}}){\models_{k}}\sigma(\for_{n}((\gamma^{0},\ldots,\gamma^{d^{\prime\prime\prime}})))$ and $(\beta^{0},\ldots,\beta^{d^{\prime}}){\models_{k}}\sigma(\for_{n}((\gamma^{0},\ldots,\gamma^{d^{\prime\prime\prime}})))$.
Consequently, $(\alpha^{0},\ldots,\alpha^{d^{\prime}}){\models_{k}}\sigma(x_{1})^{\gamma^{0}_{1}}\wedge\ldots\wedge\sigma(x_{n})^{\gamma^{0}_{n}}$ and $(\beta^{0},\ldots,\beta^{d^{\prime}}){\models_{k}}\sigma(x_{1})^{\gamma^{0}_{1}}\wedge\ldots\wedge\sigma(x_{n})^{\gamma^{0}_{n}}$.
Hence, for all $i{\in}\{1,\ldots,n\}$, $(\alpha^{0},\ldots,\alpha^{d^{\prime}}){\models_{k}}\sigma(x_{i})^{\gamma^{0}_{i}}$ and $(\beta^{0},\ldots,\beta^{d^{\prime\prime}}){\models_{k}}\sigma(x_{i})^{\gamma^{0}_{i}}$.
Thus, for all $i{\in}\{1,\ldots,n\}$, $(\alpha^{0},\ldots,\alpha^{d^{\prime}}){\models_{k}}\sigma(x_{i})$ iff $(\beta^{0},\ldots,\beta^{d^{\prime\prime}}){\models_{k}}\sigma(x_{i})$.
\end{proof}
\section{Definition of the function $f$ used in Section~\ref{new:section:Main:result}}\label{section:about:the:proof:of:the:existence:of:f:bounded:deterministic}
Let $n{\geq}1$.
Let $d{\geq}2$.
Let $(k,\sigma){\in}\SUB_{n}$ be such that $k{\geq}n$.
Let $g$ be a $d$-$(k,n)$-morphism such that for all $(\alpha^{0},\ldots,\alpha^{d^{\prime}}),(\beta^{0},\ldots,\beta^{d^{\prime\prime}}){\in}\CHA_{d}^{k}$, if $g((\alpha^{0},\ldots,\alpha^{d^{\prime}})){=}g((\beta^{0},\ldots,\beta^{d^{\prime\prime}}))$ then for all $i{\in}\{1,\ldots,n\}$, $(\alpha^{0},\ldots,\alpha^{d^{\prime}}){\models_{k}}\sigma(x_{i})$ iff $(\beta^{0},\ldots,\beta^{d^{\prime\prime}}){\models_{k}}\sigma(x_{i})$.
The proof of the existence of $g$ has been presented in Section~\ref{section:about:the:proof:of:the:existence:of:g:bounded:deterministic}.
In order to define the function $f$ used in Section~\ref{new:section:Main:result}, we need define for each $d^{\prime}{\in}\N$ such that $d^{\prime}{<}d$, a function $f_{d^{\prime}}$ from $\CHA_{=d^{\prime}}^{k}$ to $\CHA_{=d^{\prime}}^{n}$.
Firstly, we define the function $f_{0}$.
Secondly, for each $d^{\prime}{\in}\N$ such that $1{\leq}d^{\prime}{<}d$, assuming the function $f_{d^{\prime}{-}1}$ has been defined, we define the function $f_{d^{\prime}}$.
Let $U{=}\{g((\alpha^{0})):\ (\alpha^{0}){\in}\CHA_{=0}^{k}\}$.
By Proposition~\ref{simple:proposition:about:formulas:associated:to:chains}, $U{\subseteq}\CHA_{=0}^{n}$.
Let $h$ be a function from $U$ to $\CHA_{=0}^{k}$ such that for all $(\alpha^{0}){\in}\CHA_{=0}^{k}$, $g(h(g((\alpha^{0})))){=}g((\alpha^{0}))$.
Obviously, $h$ is injective.
Hence, ${\parallel}U{\parallel}{=}{\parallel}\{h(g((\alpha^{0}))):\ (\alpha^{0}){\in}\CHA_{=0}^{k}\}{\parallel}$.
Since $k{\geq}n$, therefore ${\parallel}\CHA_{=0}^{n}\setminus U{\parallel}{\leq}{\parallel}\CHA_{=0}^{k}\setminus\{h(g((\alpha^{0}))):\ (\alpha^{0}){\in}\CHA_{=0}^{k}\}{\parallel}$.
Let $S$ be a subset of $\CHA_{=0}^{k}\setminus\{h(g((\alpha^{0}))):\ (\alpha^{0}){\in}\CHA_{=0}^{k}\}$ such that ${\parallel}S{\parallel}{=}{\parallel}\CHA_{=0}^{n}\setminus U{\parallel}$.
Let $f^{\ast}_{0}$ be a one-to-one correspondence between $S$ and $\CHA_{=0}^{n}\setminus U$.
{\bf Now, we define the function $f_{0}$.}
Let $f_{0}$ be the function from $\CHA_{=0}^{k}$ to $\CHA_{=0}^{n}$ such that
\begin{itemize}
\item if $(\alpha^{0}){\in}S$ then $f_{0}((\alpha^{0})){=}f^{\ast}_{0}((\alpha^{0}))$ else $f_{0}((\alpha^{0})){=}g((\alpha^{0}))$,
\end{itemize}
where $(\alpha^{0})$ ranges over $\CHA_{=0}^{k}$.
Lemmas~\ref{first:lemma:about:f:ast:alt1} and~\ref{second:lemma:about:f:ast:alt1} show that $f_{0}$ possesses interesting properties.
%
%
\begin{lemma}\label{first:lemma:about:f:ast:alt1}
$f_{0}$ is surjective.
\end{lemma}
\begin{lemma}\label{second:lemma:about:f:ast:alt1}
For all $(\alpha^{0}),(\beta^{0}){\in}\CHA_{=0}^{k}$, if $f_{0}((\alpha^{0})){=}f_{0}((\beta^{0}))$ then $g((\alpha^{0})){=}g((\beta^{0}))$.
\end{lemma}
Let $d^{\prime}{\in}\N$ be such that $1{\leq}d^{\prime}{<}d$ and a surjective function $f_{d^{\prime}-1}$ from $\CHA_{=d^{\prime}-1}^{k}$ to $\CHA_{=d^{\prime}-1}^{n}$ has been defined such that for all $(\alpha^{1},\ldots,\alpha^{d^{\prime}}),(\beta^{1},\ldots,\beta^{d^{\prime}}){\in}\CHA_{=d^{\prime}-1}^{k}$, if $f_{d^{\prime}-1}((\alpha^{1},\ldots\alpha^{d^{\prime}})){=}f_{d^{\prime}-1}((\beta^{1},\ldots,\beta^{d^{\prime}}))$ then $g((\alpha^{1},\ldots\alpha^{d^{\prime}})){=}g((\beta^{1},\ldots,\beta^{d^{\prime}}))$.
For all $(\delta^{1},\ldots,\delta^{d^{\prime}}){\in}\CHA_{=d^{\prime}-1}^{n}$, let $S((\delta^{1},\ldots,\delta^{d^{\prime}})){=}\{(\beta^{0},\ldots,\beta^{d^{\prime}}):\ (\beta^{0},\ldots,\beta^{d^{\prime}}){\in}\CHA_{=d^{\prime}}^{k}$ and $f_{d^{\prime}{-}1}((\beta^{1},\ldots,\beta^{d^{\prime}})){=}(\delta^{1},\ldots,\delta^{d^{\prime}})\}$ and $T((\delta^{1},\ldots,\delta^{d^{\prime}})){=}\{(\epsilon^{0},\ldots,\epsilon^{d^{\prime}}):\ (\epsilon^{0},\ldots,\epsilon^{d^{\prime}}){\in}\CHA_{=d^{\prime}}^{n}$ and $(\epsilon^{1},\ldots,\epsilon^{d^{\prime}}){=}(\delta^{1},\ldots,\delta^{d^{\prime}})\}$.
For all $(\delta^{1},\ldots,\delta^{d^{\prime}}){\in}\CHA_{=d^{\prime}-1}^{n}$, let $\sim_{(\delta^{1},\ldots,\delta^{d^{\prime}})}$ be the equivalence relation on $S((\delta^{1},\ldots,\delta^{d^{\prime}}))$ such that
\begin{itemize}
\item $(\beta^{0},\ldots,\beta^{d^{\prime}}){\sim_{(\delta^{1},\ldots,\delta^{d^{\prime}})}}(\gamma^{0},\ldots,\gamma^{d^{\prime}})$ iff $g((\beta^{0},\ldots,\beta^{d^{\prime}})){=}g((\gamma^{0},\ldots,\gamma^{d^{\prime}}))$,
\end{itemize}
where $(\beta^{0},\ldots,\beta^{d^{\prime}}),(\gamma^{0},\ldots,\gamma^{d^{\prime}})$ range over $S((\delta^{1},\ldots,\delta^{d^{\prime}}))$.
The next result will allow us to use Proposition~\ref{main:one}.
\begin{proposition}\label{cardinality:ordered:appropriate:way}
For all $(\delta^{1},\ldots,\delta^{d^{\prime}}){\in}\CHA_{=d^{\prime}-1}^{n}$,
\begin{enumerate}
\item ${\parallel}S((\delta^{1},\ldots,\delta^{d^{\prime}}))/{\sim_{(\delta^{1},\ldots,\delta^{d^{\prime}})}}{\parallel}{\leq}{\parallel}T((\delta^{1},\ldots,\delta^{d^{\prime}})){\parallel}$,
\item ${\parallel}T((\delta^{1},\ldots,\delta^{d^{\prime}})){\parallel}{\leq}{\parallel}S((\delta^{1},\ldots,\delta^{d^{\prime}})){\parallel}$.
\end{enumerate}
\end{proposition}
\begin{proof}
Let $(\delta^{1},\ldots,\delta^{d^{\prime}}){\in}\CHA_{=d^{\prime}-1}^{n}$.
Obviously, ${\parallel}T((\delta^{1},\ldots,\delta^{d^{\prime}})){\parallel}{=}2^{n}$.
\\
\\
$(i)$ ---
For the sake of the contradiction, suppose ${\parallel}S((\delta^{1},\ldots,\delta^{d^{\prime}}))/{\sim_{(\delta^{1},\ldots,\delta^{d^{\prime}})}}{\parallel}{>}{\parallel}T((\delta^{1},\ldots,\delta^{d^{\prime}})){\parallel}$.
Let $p{\in}\N$ and $(\beta^{0,1},\ldots,\beta^{d^{\prime},1}),\ldots,(\beta^{0,p},\ldots,\beta^{d^{\prime},p}){\in}S((\delta^{1},\ldots,\delta^{d^{\prime}}))$ be such that $p{>}{\parallel}T((\delta^{1},\ldots,\delta^{d^{\prime}})){\parallel}$ and for all $q,r{\in}\N$, if $1{\leq}q,r{\leq}p$ and $q{\not=}r$ then $(\beta^{0,q},\ldots,\beta^{d^{\prime},q}){\not\sim_{(\delta^{1},\ldots,\delta^{d^{\prime}})}}(\beta^{0,r},\ldots,\beta^{d^{\prime},r})$.
Thus, $f_{d^{\prime}{-}1}((\beta^{1,1},\ldots,\beta^{d^{\prime},1})){=}(\delta^{1},\ldots,\delta^{d^{\prime}})$, $\ldots$, $f_{d^{\prime}{-}1}((\beta^{1,p},\ldots,\beta^{d^{\prime},p})){=}(\delta^{1},\ldots,\delta^{d^{\prime}})$.
Consequently, let $(\epsilon^{1},\ldots,\epsilon^{d^{\prime}}){\in}\CHA_{=d^{\prime}-1}^{n}$ be such that $g((\beta^{1,1},\ldots,\beta^{d^{\prime},1})){=}(\epsilon^{1},\ldots,\epsilon^{d^{\prime}})$, $\ldots$, $g((\beta^{1,p},\ldots,\beta^{d^{\prime},p})){=}(\epsilon^{1},\ldots,\epsilon^{d^{\prime}})$.
Since $g$ is a $d$-$(k,n)$-morphism, therefore let $\epsilon^{0,1},\ldots,\epsilon^{0,p}{\in}\BIT_{n}$ be such that $g((\beta^{0,1},\beta^{1,1},\ldots,\beta^{d^{\prime},1})){=}(\epsilon^{0,1},\epsilon^{1},\ldots,\epsilon^{d^{\prime}})$, $\ldots$, $g((\beta^{0,p},\beta^{1,p},\ldots,\beta^{d^{\prime},p})){=}(\epsilon^{0,p},\epsilon^{1},\ldots,\epsilon^{d^{\prime}})$.
Since for all $q,r{\in}\N$, if $1{\leq}q,r{\leq}p$ and $q{\not=}r$ then $(\beta^{0,q},\ldots,\beta^{d^{\prime},q}){\not\sim_{(\delta^{1},\ldots,\delta^{d^{\prime}})}}(\beta^{0,r},\ldots,\beta^{d^{\prime},r})$, therefore for all $q,r{\in}\N$, if $1{\leq}q,r{\leq}p$ and $q{\not=}r$ then $g((\beta^{0,q},\ldots,\beta^{d^{\prime},q}){\not=}g((\beta^{0,r},\ldots,\beta^{d^{\prime},r}))$.
Since $g((\beta^{0,1},\beta^{1,1},\ldots,\beta^{d^{\prime},1})){=}(\epsilon^{0,1},\epsilon^{1},\ldots,\epsilon^{d^{\prime}})$, $\ldots$, $g((\beta^{0,p},\beta^{1,p},\ldots,\beta^{d^{\prime},p})){=}(\epsilon^{0,p},\epsilon^{1},\ldots,\epsilon^{d^{\prime}})$, therefore for all $q,r{\in}\N$, if $1{\leq}q,r{\leq}p$ and $q{\not=}r$ then $\epsilon^{0,q}{\not=}\epsilon^{0,r}$.
Hence, $p{\leq}2^{n}$.
Since ${\parallel}T((\delta^{1},\ldots,\delta^{d^{\prime}})){\parallel}{=}2^{n}$, therefore $p{\leq}{\parallel}T((\delta^{1},\ldots,\delta^{d^{\prime}})){\parallel}$: a contradiction.
\\
\\
$(ii)$ ---
Since $f_{d^{\prime}{-}1}$ is surjective, therefore obviously, ${\parallel}S((\delta^{1},\ldots,\delta^{d^{\prime}})){\parallel}{\geq}2^{k}$.
Since $k{\geq}n$ and ${\parallel}T((\delta^{1},\ldots,\delta^{d^{\prime}})){\parallel}{=}2^{n}$, therefore ${\parallel}T((\delta^{1},\ldots,\delta^{d^{\prime}})){\parallel}{\leq}{\parallel}S((\delta^{1},\ldots,\delta^{d^{\prime}})){\parallel}$.
\end{proof}
Hence, for all $(\delta^{1},\ldots,\delta^{d^{\prime}}){\in}\CHA_{=d^{\prime}-1}^{n}$, by Propositions~\ref{main:one} and~\ref{cardinality:ordered:appropriate:way}, let $f^{(\delta^{1},\ldots,\delta^{d^{\prime}})}_{d^{\prime}}$ be a surjective function from $S((\delta^{1},\ldots,\delta^{d^{\prime}}))$ to $T((\delta^{1},\ldots,\delta^{d^{\prime}}))$ such that for all $(\beta^{0},\ldots,\beta^{d^{\prime}}),(\gamma^{0},\ldots,\gamma^{d^{\prime}}){\in}S((\delta^{1},\ldots,\delta^{d^{\prime}}))$, if $f^{(\delta^{1},\ldots,\delta^{d^{\prime}})}_{d^{\prime}}((\beta^{0},\ldots,\beta^{d^{\prime}})){=}f^{(\delta^{1},\ldots,\delta^{d^{\prime}})}_{d^{\prime}}((\gamma^{0},\ldots,\gamma^{d^{\prime}}))$ then $(\beta^{0},\ldots,\beta^{d^{\prime}}){\sim_{(\delta^{1},\ldots,\delta^{d^{\prime}})}}(\gamma^{0},\ldots,\gamma^{d^{\prime}})$.
Notice that this is the only place in the paper where we use Proposition~\ref{main:one}.
{\bf Now, we define the function $f_{d^{\prime}}$.}
Let $f_{d^{\prime}}$ be the function from $\CHA_{=d^{\prime}}^{k}$ to $\CHA_{=d^{\prime}}^{n}$ such that
\begin{itemize}
\item $f_{d^{\prime}}((\beta^{0},\ldots,\beta^{d^{\prime}})){=}f_{d^{\prime}}^{f_{d^{\prime}-1}((\beta^{1},\ldots,\beta^{d^{\prime}}))}((\beta^{0},\ldots,\beta^{d^{\prime}}))$,
\end{itemize}
where $(\beta^{0},\ldots,\beta^{d^{\prime}})$ ranges over $\CHA_{=d^{\prime}}^{k}$.
Lemmas~\ref{first:lemma:about:f:ast:alt1:bis} and~\ref{second:lemma:about:f:ast:alt1:bis} show that $f_{d^{\prime}}$ possesses interesting properties.
%
%
\begin{lemma}\label{first:lemma:about:f:ast:alt1:bis}
$f_{d^{\prime}}$ is surjective.
\end{lemma}
\begin{lemma}\label{second:lemma:about:f:ast:alt1:bis}
For all $(\alpha^{0},\ldots,\alpha^{d^{\prime}}),(\beta^{0},\ldots,\beta^{d^{\prime}}){\in}\CHA_{=d^{\prime}}^{k}$, if $f_{d^{\prime}}((\alpha^{0},\ldots,\alpha^{d^{\prime}})){=}f_{d^{\prime}}((\beta^{0},\ldots,\beta^{d^{\prime}}))$ then $g((\alpha^{0},\ldots,\alpha^{d^{\prime}})){=}g((\beta^{0},\ldots,\beta^{d^{\prime}}))$.
\end{lemma}
{\bf Now, we define the function $f$ used in Section~\ref{new:section:Main:result}.}
Let $f$ be the function from $\CHA_{d}^{k}$ to $\CHA_{d}^{n}$ such that
\begin{itemize}
\item $f((\beta^{0},\ldots,\beta^{d^{\prime}})){=}f_{d^{\prime}}((\beta^{0},\ldots,\beta^{d^{\prime}}))$,
\end{itemize}
where $(\beta^{0},\ldots,\beta^{d^{\prime}})$ ranges over $\CHA_{d}^{k}$.
Propositions~\ref{first:lemma:about:f:ast:alt1:ter}--\ref{second:lemma:about:f:ast:alt1:ter} show that $f$ possesses the properties required in Section~\ref{new:section:Main:result}.
\begin{proposition}\label{first:lemma:about:f:ast:alt1:ter}
$f$ is surjective.
\end{proposition}
\begin{proof}
Let $(\gamma^{0},\ldots,\gamma^{d^{\prime}}){\in}\CHA_{d}^{n}$.
Hence, $(\gamma^{0},\ldots,\gamma^{d^{\prime}}){\in}\CHA_{=d^{\prime}}^{n}$.
Since $f_{d^{\prime}}$ is surjective, therefore let $(\beta^{0},\ldots,\beta^{d^{\prime}}){\in}\CHA_{=d^{\prime}}^{k}$ be such that $f_{d^{\prime}}((\beta^{0},\ldots,\beta^{d^{\prime}})){=}(\gamma^{0},\ldots,\gamma^{d^{\prime}}))$.
Since $f((\beta^{0},\ldots,\beta^{d^{\prime}})){=}f_{d^{\prime}}((\beta^{0},\ldots,\beta^{d^{\prime}}))$, therefore $f((\beta^{0},\ldots,\beta^{d^{\prime}})){=}(\gamma^{0},\ldots,\gamma^{d^{\prime}}))$.
\end{proof}
\begin{proposition}\label{first:lemma:about:f:ast:alt1:ter:morphism}
$f$ is a $d$-$(k,n)$-morphism.
\end{proposition}
\begin{proof}
For the sake of the contradiction, suppose $f$ is not a $d$-$(k,n)$-morphism.
Hence, let $(\alpha^{0},\ldots,\alpha^{d^{\prime}}){\in}\CHA_{d}^{k}$ and $(\beta^{0},\ldots,\beta^{d^{\prime\prime}}){\in}\CHA_{d}^{n}$ be such that $f((\alpha^{0},\ldots,\alpha^{d^{\prime}})){=}(\beta^{0},\ldots,\beta^{d^{\prime\prime}})$ and either $d^{\prime}{\geq}1$ and $d^{\prime\prime}{=}0$, or $d^{\prime\prime}{\geq}1$ and $d^{\prime}{=}0$, or $d^{\prime}{\geq}1$, $d^{\prime\prime}{\geq}1$ and $f((\alpha^{1},\ldots,\alpha^{d^{\prime}})){\not=}(\beta^{1},\ldots,\beta^{d^{\prime\prime}})$.
Thus, $f_{d^{\prime}}((\alpha^{0},\ldots,\alpha^{d^{\prime}})){=}(\beta^{0},\ldots,\beta^{d^{\prime\prime}})$.
Since
\linebreak
$f_{d^{\prime}}((\alpha^{0},\ldots,\alpha^{d^{\prime}})){\in}\CHA_{=d^{\prime}}^{n}$ and $(\beta^{0},\ldots,\beta^{d^{\prime\prime}}){\in}\CHA_{=d^{\prime\prime}}^{n}$, therefore $d^{\prime}{=}d^{\prime\prime}$.
Since either $d^{\prime}{\geq}1$ and $d^{\prime\prime}{=}0$, or $d^{\prime\prime}{\geq}1$ and $d^{\prime}{=}0$, or $d^{\prime}{\geq}1$, $d^{\prime\prime}{\geq}1$ and $f((\alpha^{1},\ldots,\alpha^{d^{\prime}})){\not=}(\beta^{1},\ldots,\beta^{d^{\prime\prime}})$, therefore $d^{\prime}{\geq}1$, $d^{\prime\prime}{\geq}1$ and $f((\alpha^{1},\ldots,\alpha^{d^{\prime}})){\not=}(\beta^{1},\ldots,\beta^{d^{\prime\prime}})$.
Consequently, $f_{d^{\prime}-1}((\alpha^{1},\ldots,\alpha^{d^{\prime}})){\not=}(\beta^{1},\ldots,\beta^{d^{\prime\prime}})$.
Let $(\gamma^{1},\ldots,\gamma^{d^{\prime}}){\in}\CHA_{=d^{\prime}-1}^{n}$ be such that
\linebreak
$f_{d^{\prime}-1}((\alpha^{1},\ldots,\alpha^{d^{\prime}})){=}(\gamma^{1},\ldots,\gamma^{d^{\prime}})$.
Since $f_{d^{\prime}-1}((\alpha^{1},\ldots,\alpha^{d^{\prime}})){\not=}(\beta^{1},\ldots,\beta^{d^{\prime\prime}})$, therefore $(\gamma^{1},\ldots,\gamma^{d^{\prime}}){\not=}(\beta^{1},\ldots,\beta^{d^{\prime\prime}})$.
Since $d^{\prime}{\geq}1$, therefore $f_{d^{\prime}}((\alpha^{0},\ldots,\alpha^{d^{\prime}})){=}f_{d^{\prime}}^{f_{d^{\prime}-1}((\alpha^{1},\ldots,\alpha^{d^{\prime}}))}((\alpha^{0},\ldots,\alpha^{d^{\prime}}))$.
Since $f_{d^{\prime}}((\alpha^{0},\ldots,\alpha^{d^{\prime}})){=}(\beta^{0},\ldots,\beta^{d^{\prime\prime}})$ and
\linebreak
$f_{d^{\prime}-1}((\alpha^{1},\ldots,\alpha^{d^{\prime}})){=}(\gamma^{1},\ldots,\gamma^{d^{\prime}})$, therefore $f_{d^{\prime}}^{(\gamma^{1},\ldots,\gamma^{d^{\prime}})}((\alpha^{0},\ldots,\alpha^{d^{\prime}})){=}(\beta^{0},\ldots,
$\linebreak$
\beta^{d^{\prime\prime}})$.
Since $f_{d^{\prime}}^{(\gamma^{1},\ldots,\gamma^{d^{\prime}})}((\alpha^{0},\ldots,\alpha^{d^{\prime}})){\in}T((\gamma^{1},\ldots,\gamma^{d^{\prime}}))$ and $(\beta^{0},\ldots,\beta^{d^{\prime\prime}}){\in}T((\beta^{1},\ldots,\beta^{d^{\prime\prime}}))$, therefore $(\gamma^{1},\ldots,\gamma^{d^{\prime}}){=}(\beta^{1},\ldots,\beta^{d^{\prime\prime}})$: a contradiction.
\end{proof}
\begin{proposition}\label{second:lemma:about:f:ast:alt1:ter}
For all $(\alpha^{0},\ldots,\alpha^{d^{\prime}}),(\beta^{0},\ldots,\beta^{d^{\prime\prime}}){\in}\CHA_{d}^{k}$, if $f((\alpha^{0},\ldots,\alpha^{d^{\prime}})){=}f((\beta^{0},\ldots,\beta^{d^{\prime\prime}}))$ then $g((\alpha^{0},\ldots,\alpha^{d^{\prime}})){=}g((\beta^{0},\ldots,\beta^{d^{\prime\prime}}))$.
\end{proposition}
\begin{proof}
Let $(\alpha^{0},\ldots,\alpha^{d^{\prime}}),(\beta^{0},\ldots,\beta^{d^{\prime\prime}}){\in}\CHA_{d}^{k}$.
Suppose $f((\alpha^{0},\ldots,\alpha^{d^{\prime}})){=}f((\beta^{0},\ldots,\beta^{d^{\prime\prime}}))$.
Hence, $f_{d^{\prime}}((\alpha^{0},\ldots,\alpha^{d^{\prime}})){=}f_{d^{\prime\prime}}((\beta^{0},\ldots,\beta^{d^{\prime\prime}}))$.
We consider the following $2$ cases.
\\
\\
{\bf Case $d^{\prime}{=}0$ and $d^{\prime\prime}{=}0$.}
Since $f_{d^{\prime}}((\alpha^{0},\ldots,\alpha^{d^{\prime}})){=}f_{d^{\prime\prime}}((\beta^{0},\ldots,\beta^{d^{\prime\prime}}))$, $d^{\prime}{=}0$ and $d^{\prime\prime}{=}0$, therefore by Lemma~\ref{second:lemma:about:f:ast:alt1}, $g((\alpha^{0},\ldots,\alpha^{d^{\prime}})){=}g((\beta^{0},\ldots,\beta^{d^{\prime}}))$.
\\
\\
{\bf Case either $d^{\prime}{\geq}1$, or $d^{\prime\prime}{\geq}1$.}
Since $f_{d^{\prime}}((\alpha^{0},\ldots,\alpha^{d^{\prime}})){=}f_{d^{\prime\prime}}((\beta^{0},\ldots,\beta^{d^{\prime\prime}}))$, $f_{d^{\prime}}((\alpha^{0},
\linebreak
\ldots,\alpha^{d^{\prime}})){\in}\CHA_{=d^{\prime}}^{n}$ and $f_{d^{\prime\prime}}((\beta^{0},\ldots,\beta^{d^{\prime\prime}})){\in}\CHA_{=d^{\prime\prime}}^{n}$, therefore $d^{\prime}{=}d^{\prime\prime}$.
Since$f_{d^{\prime}}((\alpha^{0},\ldots,\alpha^{d^{\prime}})){=}f_{d^{\prime\prime}}((\beta^{0},\ldots,\beta^{d^{\prime\prime}}))$, therefore by Lemma~\ref{second:lemma:about:f:ast:alt1:bis}, $g((\alpha^{0},\ldots,\alpha^{d^{\prime}})){=}g((\beta^{0},\ldots,\beta^{d^{\prime\prime}}))$.
\end{proof}
\section{Conclusion}
In this paper, we have proved that for all $d{\geq}2$, in $\Alt_{1}+\square^{d}\bot$, unifiable formulas are unitary for elementary unification.
Here are open questions concerning unification types for elementary unification and unification with constants:
\begin{description}
\item[$1.$] determine for all $d{\geq}2$, the unification type of the locally tabular modal logic $\K+\square^{d}\bot$ for elementary unification,
\item[$2.$] determine for all $d{\geq}2$, the unification type of the locally tabular modal logics $\Alt_{1}+\square^{d}\bot$ and $\K+\square^{d}\bot$ for unification with constants,
\item[$3.$] determine the unification type of other locally tabular modal logics like the ones studied in~\cite{Miyazaki:2004,Nagle:Thomason:1985,Shapirovsky:Shehtman:2016} for elementary unification and unification with constants,
\item[$4.$] determine the unification type of $\Alt_{1}+\lozenge\top$ for elementary unification.
\end{description}
We conjecture that the modal logics mentioned in Items~$1$--$3$ are either finitary, or unitary within the corresponding considered contexts of unification.
As for the unification type within the context of elementary unification considered in Item~$4$, it is still a mystery.
\\
\\
On the side of computability and complexity, it is known that elementary unification is in $\PSPACE$ for $\Alt_{1}$~\cite{Balbiani:Tinchev:2016} and decidable for $\LTL$~\cite{Rybakov:2008}.
As for $\Alt_{1}+\lozenge\top$, the membership in $\NP$ of its elementary unification problem is a direct consequence of the fact that in this modal logic, one can easily determine if a given variable-free formula is equivalent to $\bot$, or is equivalent to $\top$.
Here are open questions concerning the computability and the complexity of elementary unification and unification with constants:
\begin{description}
\item[$5.$] determine for all $d{\geq}2$, the complexity of elementary unification for the locally tabular modal logics $\Alt_{1}+\square^{d}\bot$ and $\K+\square^{d}\bot$,
\item[$6.$] determine for all $d{\geq}2$, the complexity of unification with constants for the locally tabular modal logics $\Alt_{1}+\square^{d}\bot$ and $\K+\square^{d}\bot$,
\item[$7.$] determine the complexity of elementary unification and unification with constants for other locally tabular modal logics like the ones studied in~\cite{Miyazaki:2004,Nagle:Thomason:1985,Shapirovsky:Shehtman:2016},
\item[$8.$] determine the computability of unification with constants for $\LTL$ and $\Alt_{1}+\lozenge\top$.
\end{description}
The local tabularity of the modal logics mentioned in Items~$5$--$7$ implies the decidability of the corresponding considered unification problems.
As for the unification problems with constants considered in Item~$8$, its computability is still a mystery.
\section*{Acknowledgements}
The preparation of this paper has been supported by {\em Bulgarian Science Fund}\/ (Project \emph{DN02/15/19.12.2016}) and {\em Universit\'e Paul Sabatier}\/ (Programme \emph{Professeurs invit\'es 2018}).
%
%
%
%
\bibliographystyle{named}
\section*{Appendix}
{\bf Proof of Proposition~\ref{main:one}.}
$(i)\Rightarrow(ii)$ ---
Suppose ${\parallel}S/{\sim}{\parallel}{\leq}{\parallel}T{\parallel}{\leq}{\parallel}S{\parallel}$.
Let $h$ be a function from $S/{\sim}$ to $S$ such that for all $\alpha{\in}S$, $h(\lbrack\alpha\rbrack){\in}\lbrack\alpha\rbrack$.
Obviously, $h$ is injective.
Let $S_{0}{=}\{h(\lbrack\alpha\rbrack):\ \alpha{\in}S\}$.
Since $h$ is injective, therefore ${\parallel}S/{\sim}{\parallel}{=}{\parallel}S_{0}{\parallel}$.
Since ${\parallel}S/{\sim}{\parallel}{\leq}{\parallel}T{\parallel}$, therefore ${\parallel}S_{0}{\parallel}{\leq}{\parallel}T{\parallel}$.
Let $T_{0}$ be a subset of $T$ such that ${\parallel}T_{0}{\parallel}{=}{\parallel}S_{0}{\parallel}$.
Let $f_{0}$ be a one-to-one correspondence between $S_{0}$ and $T_{0}$.
Let $T_{1}{=}T\backslash T_{0}$.
Notice that $T_{0}$ and $T_{1}$ make a partition of $T$.
Since ${\parallel}T{\parallel}{\leq}{\parallel}S{\parallel}$ and ${\parallel}T_{0}{\parallel}{=}{\parallel}S_{0}{\parallel}$, therefore ${\parallel}T_{1}{\parallel}{\leq}{\parallel}S\backslash S_{0}{\parallel}$.
Let $S_{1}$ be a subset of $S\backslash S_{0}$ such that ${\parallel}S_{1}{\parallel}{=}{\parallel}T_{1}{\parallel}$.
Let $f_{1}$ be a one-to-one correspondence between $S_{1}$ and $T_{1}$.
Let $S_{2}{=}(S\backslash S_{0})\backslash S_{1}$.
Let $f_{2}$ be the function from $S_{2}$ to $T$ such that for all $\alpha{\in}S_{2}$, $f_{2}(\alpha){=}f_{0}(h(\lbrack\alpha\rbrack))$.
Let $f$ be the function from $S$ to $T$ defined by $f\arrowvert S_{0}{=}f_{0}$, $f\arrowvert S_{1}{=}f_{1}$ and $f\arrowvert S_{2}{=}f_{2}$.
%
%
\\
\\
{\bf Lemma}
$f$ is surjective.
%
%
\\
\\
{\bf Proof.}
Let $\beta{\in}T$.
We consider the following $2$ cases.
\\
\\
{\em Case $\beta{\in}T_{0}$.}
Since $f_{0}$ is one-to-one, therefore let $\alpha{\in}S_{0}$ be such that $f_{0}(\alpha){=}\beta$.
Thus, $\alpha{\in}S$.
Moreover, $f(\alpha){=}f_{0}(\alpha)$.
Since $f_{0}(\alpha){=}\beta$, therefore $f(\alpha){=}\beta$.
\\
\\
{\em Case $\beta{\in}T_{1}$.}
Since $f_{1}$ is one-to-one, therefore let $\alpha{\in}S_{1}$ be such that $f_{1}(\alpha){=}\beta$.
Hence, $\alpha{\in}S$.
Moreover, $f(\alpha){=}f_{1}(\alpha)$.
Since $f_{1}(\alpha){=}\beta$, therefore $f(\alpha){=}\beta$.
%
%
\\
\\
{\bf Lemma}
For all $\alpha,\beta{\in}S$, if $f(\alpha){=}f(\beta)$ then $\alpha{\sim}\beta$.
%
%
\\
\\
{\bf Proof.}
Let $\alpha,\beta{\in}S$ be such that $f(\alpha){=}f(\beta)$.
We consider the following $6$ cases.
\\
\\
{\em Case $\alpha{\in}S_{0}$ and $\beta{\in}S_{0}$.}
Consequently, $f(\alpha){=}f_{0}(\alpha)$ and $f(\beta){=}f_{0}(\beta)$.
Since $f(\alpha){=}f(\beta)$, therefore $f_{0}(\alpha){=}f_{0}(\beta)$.
Since $f_{0}$ is one-to-one, therefore $\alpha{=}\beta$.
Thus, $\alpha{\sim}\beta$.
\\
\\
{\em Case $\alpha{\in}S_{0}$ and $\beta{\in}S_{1}$.}
Consequently, $f(\alpha){=}f_{0}(\alpha)$ and $f(\beta){=}f_{1}(\beta)$.
Since $f(\alpha){=}f(\beta)$, therefore $f_{0}(\alpha){=}f_{1}(\beta)$.
Since $f_{0}(\alpha){\in}T_{0}$ and $f_{1}(\beta){\in}T_{1}$, therefore $T_{0}$ and $T_{1}$ do not make a partition of $T$: a contradiction.
\\
\\
{\em Case $\alpha{\in}S_{0}$ and $\beta{\in}S_{2}$.}
Hence, $f(\alpha){=}f_{0}(\alpha)$ and $f(\beta){=}f_{2}(\beta)$.
Since $f(\alpha){=}f(\beta)$, therefore $f_{0}(\alpha){=}f_{2}(\beta)$.
Thus, $f_{0}(\alpha){=}f_{0}(h(\lbrack\beta\rbrack))$.
Since $f_{0}$ is one-to-one, therefore $\alpha{=}h(\lbrack\beta\rbrack)$.
Since $h(\lbrack\beta\rbrack){\in}\lbrack\beta\rbrack$, therefore $\alpha{\in}\lbrack\beta\rbrack$.
Consequently, $\alpha{\sim}\beta$.
\\
\\
{\em Case $\alpha{\in}S_{1}$ and $\beta{\in}S_{1}$.}
Hence, $f(\alpha){=}f_{1}(\alpha)$ and $f(\beta){=}f_{1}(\beta)$.
Since $f(\alpha){=}f(\beta)$, therefore $f_{1}(\alpha){=}f_{1}(\beta)$.
Since $f_{1}$ is one-to-one, therefore $\alpha{=}\beta$.
Thus, $\alpha{\sim}\beta$.
\\
\\
{\em Case $\alpha{\in}S_{1}$ and $\beta{\in}S_{2}$.}
Hence, $f(\alpha){=}f_{1}(\alpha)$ and $f(\beta){=}f_{2}(\beta)$.
Since $f(\alpha){=}f(\beta)$, therefore $f_{1}(\alpha){=}f_{2}(\beta)$.
Thus, $f_{1}(\alpha){=}f_{0}(h(\lbrack\beta\rbrack))$.
Since $f_{1}(\alpha){\in}T_{1}$ and $f_{0}(h(\lbrack\beta\rbrack)){\in}T_{0}$, therefore $T_{0}$ and $T_{1}$ do not make a partition of $T$: a contradiction.
\\
\\
{\em Case $\alpha{\in}S_{2}$ and $\beta{\in}S_{2}$.}
Hence, $f(\alpha){=}f_{2}(\alpha)$ and $f(\beta){=}f_{2}(\beta)$.
Since $f(\alpha){=}f(\beta)$, therefore $f_{2}(\alpha){=}f_{2}(\beta)$.
Consequently, $f_{0}(h(\lbrack\alpha\rbrack)){=}f_{0}(h(\lbrack\beta\rbrack))$.
Since $f_{0}$ is one-to-one, therefore $h(\lbrack\alpha\rbrack){=}h(\lbrack\beta\rbrack)$.
Since $h(\lbrack\alpha\rbrack){\in}\lbrack\alpha\rbrack$ and $h(\lbrack\beta\rbrack){\in}\lbrack\beta\rbrack$, therefore $\lbrack\alpha\rbrack\cap\lbrack\beta\rbrack{\not=}\emptyset$.
Thus, $\alpha{\sim}\beta$.
\\
\\
$(ii)\Rightarrow(i)$ ---
Suppose $f$ is a surjective function from $S$ to $T$ such that for all $\alpha,\beta{\in}S$, if $f(\alpha){=}f(\beta)$ then $\alpha{\sim}\beta$.
For the sake of the contradiction, suppose either ${\parallel}S/{\sim}{\parallel}{>}{\parallel}T{\parallel}$, or ${\parallel}T{\parallel}{>}{\parallel}S{\parallel}$.
Since $f$ is surjective, therefore ${\parallel}T{\parallel}{\leq}{\parallel}S{\parallel}$.
Since either ${\parallel}S/{\sim}{\parallel}{>}{\parallel}T{\parallel}$, or ${\parallel}T{\parallel}{>}{\parallel}S{\parallel}$, therefore ${\parallel}S/{\sim}{\parallel}{>}{\parallel}T{\parallel}$.
Let $p{\in}\N$ and $\beta^{1},\ldots,\beta^{p}{\in}S$ be such that $p{>}{\parallel}T{\parallel}$ and for all $q,r{\in}\N$, if $1{\leq}q,r{\leq}p$ and $q{\not=}r$ then $\beta^{q}{\not\sim}\beta^{r}$.
Hence, for all $q,r{\in}\N$, if $1{\leq}q,r{\leq}p$ and $q{\not=}r$ then $f(\beta^{q}){\not=}f(\beta^{r})$.
Thus, $p{\leq}{\parallel}T{\parallel}$: a contradiction.
\\
\\
\\
\\
{\bf Proof of Lemma~\ref{lemma:6:about:psi:beta:B:equivalent:condi}.}
By induction on $\psi{\in}\FOR_{n}$.
We only consider the following $2$ cases.
\\
\\
{\bf Case $\psi{=}x_{i}$.}
Let $(\beta^{0},\ldots,\beta^{d^{\prime\prime}}){\in}\CHA_{d}^{n}$.
\begin{description}
\item[$(i)\Rightarrow(ii)$ ---] Suppose $(\alpha^{0},\ldots,\alpha^{d^{\prime}}){\in}\CHA_{d}^{k}$ is such that $f((\alpha^{0},\ldots,\alpha^{d^{\prime}})){=}(\beta^{0},\ldots,\beta^{d^{\prime\prime}})$ and $(\alpha^{0},\ldots,\alpha^{d^{\prime}}){\models_{k}}\sigma(x_{i})$.
Let $(\gamma^{0},\ldots,\gamma^{d^{\prime\prime\prime}}){\in}\CHA_{d}^{k}$ be such that $f((\gamma^{0},\ldots,\gamma^{d^{\prime\prime\prime}})){=}(\beta^{0},\ldots,\beta^{d^{\prime\prime}})$.
Since $f((\alpha^{0},\ldots,\alpha^{d^{\prime}})){=}(\beta^{0},\ldots,\beta^{d^{\prime\prime}})$, therefore $f((\alpha^{0},\ldots,\alpha^{d^{\prime}})){=}f((\gamma^{0},\ldots,\gamma^{d^{\prime\prime\prime}}))$.
Hence, $g((\alpha^{0},\ldots,\alpha^{d^{\prime}})){=}g((\gamma^{0},\ldots,\gamma^{d^{\prime\prime\prime}}))$.
Thus, $(\alpha^{0},\ldots,\alpha^{d^{\prime}}){\models_{k}}\sigma(x_{i})$ iff $(\gamma^{0},\ldots,\gamma^{d^{\prime\prime\prime}}){\models_{k}}\sigma(x_{i})$.
Since $(\alpha^{0},\ldots,\alpha^{d^{\prime}}){\models_{k}}\sigma(x_{i})$, therefore $(\gamma^{0},\ldots,\gamma^{d^{\prime\prime\prime}}){\models_{k}}\sigma(x_{i})$.
\item[$(ii)\Rightarrow(iii)$ ---] Suppose for all $(\alpha^{0},\ldots,\alpha^{d^{\prime}}){\in}\CHA_{d}^{k}$, if $f((\alpha^{0},\ldots,\alpha^{d^{\prime}})){=}(\beta^{0},\ldots,\beta^{d^{\prime\prime}})$ then $(\alpha^{0},\ldots,\alpha^{d^{\prime}}){\models_{k}}\sigma(x_{i})$.
Since $f$ is surjective, therefore let $(\gamma^{0},\ldots,\gamma^{d^{\prime\prime\prime}}){\in}\CHA_{d}^{k}$ be such that $f((\gamma^{0},\ldots,\gamma^{d^{\prime\prime\prime}})){=}(\beta^{0},\ldots,\beta^{d^{\prime\prime}})$.
Since for all $(\alpha^{0},\ldots,\alpha^{d^{\prime}}){\in}\CHA_{d}^{k}$, if $f((\alpha^{0},\ldots,\alpha^{d^{\prime}})){=}(\beta^{0},\ldots,\beta^{d^{\prime\prime}})$ then $(\alpha^{0},\ldots,\alpha^{d^{\prime}}){\models_{k}}\sigma(x_{i})$, therefore $(\gamma^{0},\ldots,\gamma^{d^{\prime\prime\prime}}){\models_{k}}\sigma(x_{i})$.
Consequently, $(\beta^{0},\ldots,\beta^{d^{\prime\prime}}){\models_{n}}\for_{n}(f((\gamma^{0},\ldots,\gamma^{d^{\prime\prime\prime}})))\rightarrow\tau(x_{i})$.
Since$f((\gamma^{0},\ldots,\gamma^{d^{\prime\prime\prime}})){=}(\beta^{0},\ldots,\beta^{d^{\prime\prime}})$, therefore $(\beta^{0},\ldots,\beta^{d^{\prime\prime}}){\models_{n}}\for_{n}((\beta^{0},\ldots,\beta^{d^{\prime\prime}}))\rightarrow\tau(x_{i})$.
Since by Proposition~\ref{lemma:1:about:alpha:beta:equivalent:conditions}, $(\beta^{0},\ldots,\beta^{d^{\prime\prime}})\models_{n}\for_{n}((\beta^{0},\ldots,\beta^{d^{\prime\prime}}))$, therefore $(\beta^{0},\ldots,\beta^{d^{\prime\prime}})\models_{n}\tau(x_{i})$.
\item[$(iii)\Rightarrow(i)$ ---] Suppose $(\beta^{0},\ldots,\beta^{d^{\prime\prime}}){\models_{n}}\tau(x_{i})$.
Let $(\alpha^{0},\ldots,\alpha^{d^{\prime}}){\in}\CHA_{d}^{k}$ be such that $(\alpha^{0},\ldots,\alpha^{d^{\prime}}){\models_{k}}\sigma(x_{i})$ and $(\beta^{0},\ldots,\beta^{d^{\prime\prime}}){\models_{n}}\for_{n}(f((\alpha^{0},\ldots,\alpha^{d^{\prime}})))$.
Hence, by Proposition~\ref{lemma:1:about:alpha:beta:equivalent:conditions}, $f((\alpha^{0},\ldots,\alpha^{d^{\prime}})){=}(\beta^{0},\ldots,\beta^{d^{\prime\prime}})$.
\end{description}
{\bf Case $\psi{=}\square\chi$.}
Let $(\beta^{0},\ldots,\beta^{d^{\prime\prime}}){\in}\CHA_{d}^{n}$.
\begin{description}
\item[$(i)\Rightarrow(ii)$ ---] Suppose $(\alpha^{0},\ldots,\alpha^{d^{\prime}}){\in}\CHA_{d}^{k}$ is such that $f((\alpha^{0},\ldots,\alpha^{d^{\prime}})){=}(\beta^{0},\ldots,\beta^{d^{\prime\prime}})$ and $(\alpha^{0},\ldots,\alpha^{d^{\prime}}){\models_{k}}\sigma(\square\chi)$.
Let $(\gamma^{0},\ldots,\gamma^{d^{\prime\prime\prime}}){\in}\CHA_{d}^{k}$ be such that $f((\gamma^{0},\ldots,\gamma^{d^{\prime\prime\prime}})){=}(\beta^{0},\ldots,\beta^{d^{\prime\prime}})$.
Suppose $(\gamma^{0},\ldots,\gamma^{d^{\prime\prime\prime}}){\not\models_{k}}\sigma(\square\chi)$.
Thus, $d^{\prime\prime\prime}{\geq}1$ and $(\gamma^{1},\ldots,\gamma^{d^{\prime\prime\prime}}){\not\models_{k}}\sigma(\chi)$.
Since $f$ is a $d$-$(k,n)$-morphism and $f((\gamma^{0},\ldots,\gamma^{d^{\prime\prime\prime}})){=}(\beta^{0},\ldots,\beta^{d^{\prime\prime}})$, therefore $d^{\prime\prime}{\geq}1$ and $f((\gamma^{1},\ldots,\gamma^{d^{\prime\prime\prime}})){=}(\beta^{1},\ldots,\beta^{d^{\prime\prime}})$.
Since $f$ is a $d$-$(k,n)$-morphism and $f((\alpha^{0},\ldots,\alpha^{d^{\prime}})){=}(\beta^{0},\ldots,\beta^{d^{\prime\prime}})$, therefore $d^{\prime}{\geq}1$ and $f((\alpha^{1},\ldots,\alpha^{d^{\prime}})){=}(\beta^{1},\ldots,\beta^{d^{\prime\prime}})$.
Since $(\gamma^{1},\ldots,\gamma^{d^{\prime\prime\prime}}){\not\models_{k}}\sigma(\chi)$ and $f((\gamma^{1},\ldots,\gamma^{d^{\prime\prime\prime}})){=}(\beta^{1},\ldots,\beta^{d^{\prime\prime}})$, therefore by induction hypothesis, $(\alpha^{1},\ldots,\alpha^{d^{\prime}}){\not\models_{k}}\sigma(\chi)$.
Hence, $(\alpha^{0},\ldots,\alpha^{d^{\prime}}){\not\models_{k}}\sigma(\square\chi)$: a contradiction.
\item[$(ii)\Rightarrow(iii)$ ---] Suppose for all $(\alpha^{0},\ldots,\alpha^{d^{\prime}}){\in}\CHA_{d}^{k}$, if $f((\alpha^{0},\ldots,\alpha^{d^{\prime}})){=}(\beta^{0},\ldots,\beta^{d^{\prime\prime}})$ then $(\alpha^{0},\ldots,\alpha^{d^{\prime}}){\models_{k}}\sigma(\square\chi)$.
Suppose $(\beta^{0},\ldots,\beta^{d^{\prime\prime}}){\not\models_{n}}\tau(\square\chi)$.
Consequently, $d^{\prime\prime}{\geq}1$ and $(\beta^{1},\ldots,\beta^{d^{\prime\prime}}){\not\models_{n}}\tau(\chi)$.
Since $f$ is surjective, therefore let $(\gamma^{0},\ldots,\gamma^{d^{\prime\prime\prime}}){\in}\CHA_{d}^{k}$ be such that $f((\gamma^{0},\ldots,\gamma^{d^{\prime\prime\prime}})){=}(\beta^{0},\ldots,\beta^{d^{\prime\prime}})$.
Since for all $(\alpha^{0},\ldots,\alpha^{d^{\prime}}){\in}\CHA_{d}^{k}$, if $f((\alpha^{0},\ldots,\alpha^{d^{\prime}})){=}(\beta^{0},\ldots,\beta^{d^{\prime\prime}})$ then $(\alpha^{0},\ldots,\alpha^{d^{\prime}}){\models_{k}}\sigma(\square\chi)$, therefore $(\gamma^{0},\ldots,\gamma^{d^{\prime\prime\prime}}){\models_{k}}\sigma(\square\chi)$.
Since $f$ is a $d$-$(k,n)$-morphism, $d^{\prime\prime}{\geq}1$ and $f((\gamma^{0},\ldots,\gamma^{d^{\prime\prime\prime}})){=}(\beta^{0},\ldots,\beta^{d^{\prime\prime}})$, therefore $d^{\prime\prime\prime}{\geq}1$ and $f((\gamma^{1},\ldots,\gamma^{d^{\prime\prime\prime}})){=}(\beta^{1},\ldots,\beta^{d^{\prime\prime}})$.
Since $(\beta^{1},\ldots,\beta^{d^{\prime\prime}}){\not\models_{n}}\tau(\chi)$, therefore by induction hypothesis, $(\gamma^{1},\ldots,\gamma^{d^{\prime\prime\prime}}){\not\models_{k}}\sigma(\chi)$.
Thus, $(\gamma^{0},\ldots,\gamma^{d^{\prime\prime\prime}}){\not\models_{k}}\sigma(\square\chi)$: a contradiction.
\item[$(iii)\Rightarrow(i)$ ---] Suppose $(\beta^{0},\ldots,\beta^{d^{\prime\prime}}){\models_{n}}\tau(\square\chi)$.
Since $f$ is surjective, therefore let $(\alpha^{0},\ldots,\alpha^{d^{\prime}}){\in}\CHA_{d}^{k}$ be such that $f((\alpha^{0},\ldots,\alpha^{d^{\prime}})){=}(\beta^{0},\ldots,\beta^{d^{\prime\prime}})$.
Suppose $(\alpha^{0},\ldots,\alpha^{d^{\prime}}){\not\models_{k}}\sigma(\square\chi)$.
Consequently, $d^{\prime}{\geq}1$ and $(\alpha^{1},\ldots,\alpha^{d^{\prime}}){\not\models_{k}}\sigma(\chi)$.
Since $f$ is a $d$-$(k,n)$-morphism and $f((\alpha^{0},\ldots,\alpha^{d^{\prime}})){=}(\beta^{0},\ldots,\beta^{d^{\prime\prime}})$, therefore $d^{\prime\prime}{\geq}1$ and $f((\alpha^{1},\ldots,\alpha^{d^{\prime}})){=}(\beta^{1},\ldots,\beta^{d^{\prime\prime}})$.
Since $(\alpha^{1},\ldots,\alpha^{d^{\prime}}){\not\models_{k}}\sigma(\chi)$, therefore by induction hypothesis, $(\beta^{1},\ldots,\beta^{d^{\prime\prime}}){\not\models_{n}}\tau(\chi)$.
Hence, $(\beta^{0},\ldots,\beta^{d^{\prime\prime}}){\not\models_{n}}\tau(\square\chi)$: a contradiction.
\end{description}
\vspace{0.93cm}
{\bf Proof of Lemma~\ref{lemma:7:about:beta:B:i:following:equivalent:conditions}.}
Let $(\beta^{0},\ldots,\beta^{d^{\prime\prime}}){\in}\CHA_{d}^{k}$ and $i{\in}\{1,\ldots,n\}$.
\\
\\
$(i)\Rightarrow(ii)$ ---
Suppose $(\beta^{0},\ldots,\beta^{d^{\prime\prime}}){\models_{k}}\nu(x_{i})$.
Let $(\alpha^{0},\ldots,\alpha^{d^{\prime}}){\in}\CHA_{d}^{k}$ be such that $f((\alpha^{0},\ldots,\alpha^{d^{\prime}})){\models_{n}}x_{i}$ and $(\beta^{0},\ldots,\beta^{d^{\prime\prime}}){\models_{k}}\for_{k}((\alpha^{0},\ldots,\alpha^{d^{\prime}}))$.
Thus, by Proposition~\ref{lemma:1:about:alpha:beta:equivalent:conditions}, $(\beta^{0},\ldots,\beta^{d^{\prime\prime}}){=}(\alpha^{0},\ldots,\alpha^{d^{\prime}})$.
Since $f((\alpha^{0},\ldots,\alpha^{d^{\prime}})){\models_{n}}x_{i}$, therefore $f((\beta^{0},\ldots,\beta^{d^{\prime\prime}}))\models_{n}x_{i}$.
\\
\\
$(ii)\Rightarrow(i)$ ---
Suppose $f((\beta^{0},\ldots,\beta^{d^{\prime\prime}})){\models_{n}}x_{i}$.
Consequently, $(\beta^{0},\ldots,\beta^{d^{\prime\prime}}){\models_{k}}\for_{k}((\beta^{0},\ldots,\beta^{d^{\prime\prime}}))\rightarrow\nu(x_{i})$.
Since by Proposition~\ref{lemma:1:about:alpha:beta:equivalent:conditions}, $(\beta^{0},\ldots,\beta^{d^{\prime\prime}}){\models_{k}}\for_{k}((\beta^{0},\ldots,\beta^{d^{\prime\prime}}))$, therefore $(\beta^{0},\ldots,\beta^{d^{\prime\prime}}){\models_{k}}\nu(x_{i})$.
\\
\\
\\
\\
{\bf Proof of Lemma~\ref{lemma:8:about:beta:B:MOD:gamma:C:are:equi}.}
By $\ll$-induction on $(d^{\prime\prime},d^{\prime\prime\prime})$.
We consider the following $2$ cases.
\\
\\
{\bf Case $d^{\prime\prime}{=}0$ and $d^{\prime\prime\prime}{=}0$.}
\begin{description}
\item[$(i)\Rightarrow(ii)$ ---] Suppose $f((\beta^{0},\ldots,\beta^{d^{\prime\prime}})){=}(\gamma^{0},\ldots,\gamma^{d^{\prime\prime\prime}})$.
Since for all $i{\in}\{1,\ldots,n\}$, $(\gamma^{0},\ldots,\gamma^{d^{\prime\prime\prime}}){\models_{n}}x_{i}^{\gamma^{0}_{i}}$, therefore for all $i{\in}\{1,\ldots,n\}$, $f((\beta^{0},\ldots,\beta^{d^{\prime\prime}})){\models_{n}}x_{i}^{\gamma^{0}_{i}}$.
Thus, for all $i{\in}\{1,\ldots,n\}$, by Lemma~\ref{lemma:7:about:beta:B:i:following:equivalent:conditions}, $(\beta^{0},\ldots,\beta^{d^{\prime\prime}}){\models_{k}}\nu(x_{i})^{\gamma^{0}_{i}}$.
Hence, $(\beta^{0},\ldots,\beta^{d^{\prime\prime}}){\models_{k}}\nu(x_{1})^{\gamma^{0}_{1}}\wedge\ldots\wedge\nu(x_{n})^{\gamma^{0}_{n}}$.
Since $d^{\prime\prime}{=}0$ and $d^{\prime\prime\prime}{=}0$, therefore $(\beta^{0},\ldots,\beta^{d^{\prime\prime}}){\models_{k}}\nu(\for_{n}((\gamma^{0},\ldots,\gamma^{d^{\prime\prime\prime}})))$.
\item[$(ii)\Rightarrow(i)$ ---] Suppose $(\beta^{0},\ldots,\beta^{d^{\prime\prime}}){\models_{k}}\nu(\for_{n}((\gamma^{0},\ldots,\gamma^{d^{\prime\prime\prime}})))$.
Consequently, $(\beta^{0},\ldots,\beta^{d^{\prime\prime}}){\models_{k}}\nu(x_{1})^{\gamma^{0}_{1}}\wedge\ldots\wedge\nu(x_{n})^{\gamma^{0}_{n}}$.
Hence, for all $i{\in}\{1,\ldots,n\}$, $(\beta^{0},\ldots,\beta^{d^{\prime\prime}}){\models_{k}}\nu(x_{i})^{\gamma^{0}_{i}}$.
Thus, for all $i{\in}\{1,\ldots,n\}$, by Lemma~\ref{lemma:7:about:beta:B:i:following:equivalent:conditions}, $f((\beta^{0},\ldots,\beta^{d^{\prime\prime}})){\models_{n}}x_{i}^{\gamma^{0}_{i}}$.
Since $d^{\prime\prime}{=}0$ and $d^{\prime\prime\prime}{=}0$, therefore by Proposition~\ref{simple:proposition:about:morphisms}, $f((\beta^{0},\ldots,\beta^{d^{\prime\prime}})){=}(\gamma^{0},\ldots,\gamma^{d^{\prime\prime\prime}})$.
\end{description}
{\bf Case either $d^{\prime\prime}{\geq}1$, or $d^{\prime\prime\prime}{\geq}1$.}
\begin{description}
\item[$(i)\Rightarrow(ii)$ ---] Suppose $f((\beta^{0},\ldots,\beta^{d^{\prime\prime}})){=}(\gamma^{0},\ldots,\gamma^{d^{\prime\prime\prime}})$.
Since for all $i{\in}\{1,\ldots,n\}$, $(\gamma^{0},\ldots,\gamma^{d^{\prime\prime\prime}}){\models_{n}}x_{i}^{\gamma^{0}_{i}}$, therefore for all $i{\in}\{1,\ldots,n\}$, $f((\beta^{0},\ldots,\beta^{d^{\prime\prime}})){\models_{n}}x_{i}^{\gamma^{0}_{i}}$.
Moreover, since $f$ is a $d$-$(k,n)$-morphism and either $d^{\prime\prime}{\geq}1$, or $d^{\prime\prime\prime}{\geq}1$, therefore $d^{\prime\prime}{\geq}1$, $d^{\prime\prime\prime}{\geq}1$ and $f((\beta^{1},\ldots,\beta^{d^{\prime\prime}})){=}(\gamma^{1},\ldots,\gamma^{d^{\prime\prime\prime}})$.
Hence, for all $i{\in}\{1,\ldots,n\}$, by Lemma~\ref{lemma:7:about:beta:B:i:following:equivalent:conditions}, $(\beta^{0},\ldots,\beta^{d^{\prime\prime}}){\models_{k}}\nu(x_{i})^{\gamma^{0}_{i}}$.
Moreover, by induction hypothesis, $(\beta^{1},\ldots,\beta^{d^{\prime\prime}}){\models_{k}}\nu(\for_{n}((\gamma^{1},\ldots,\gamma^{d^{\prime\prime\prime}})))$.
Consequently, $(\beta^{0},\ldots,\beta^{d^{\prime\prime}}){\models_{k}}\nu(x_{1})^{\gamma^{0}_{1}}\wedge\ldots\wedge\nu(x_{n})^{\gamma^{0}_{n}}$.
Moreover, $(\beta^{0},\ldots,\beta^{d^{\prime\prime}}){\models_{k}}\lozenge\nu(\for_{n}((\gamma^{1},\ldots,\gamma^{d^{\prime\prime\prime}})))$.
Thus, $(\beta^{0},\ldots,\beta^{d^{\prime\prime}}){\models_{k}}\nu(\for_{n}((\gamma^{0},\ldots,\gamma^{d^{\prime\prime\prime}})))$.
\item[$(ii)\Rightarrow(i)$ ---] Suppose $(\beta^{0},\ldots,\beta^{d^{\prime\prime}}){\models_{k}}\nu(\for_{n}((\gamma^{0},\ldots,\gamma^{d^{\prime\prime\prime}})))$.
Hence, if $d^{\prime\prime\prime}{\geq}1$ then $(\beta^{0},\ldots,\beta^{d^{\prime\prime}}){\models_{k}}\nu(x_{1})^{\gamma^{0}_{1}}\wedge\ldots\wedge\nu(x_{n})^{\gamma^{0}_{n}}\wedge\lozenge\nu(\for_{n}((\gamma^{1},\ldots,\gamma^{d^{\prime\prime\prime}})))$ else $(\beta^{0},\ldots,\beta^{d^{\prime\prime}}){\models_{k}}\nu(x_{1})^{\gamma^{0}_{1}}\wedge\ldots\wedge\nu(x_{n})^{\gamma^{0}_{n}}\wedge\square\bot$.
Since either $d^{\prime\prime}{\geq}1$, or $d^{\prime\prime\prime}{\geq}1$, therefore $d^{\prime\prime}{\geq}1$, $d^{\prime\prime\prime}{\geq}1$ and for all $i{\in}\{1,\ldots,n\}$, $(\beta^{0},\ldots,\beta^{d^{\prime\prime}}){\models_{k}}\nu(x_{i})^{\gamma^{0}_{i}}$.
Moreover, $(\beta^{1},\ldots,\beta^{d^{\prime\prime}}){\models_{k}}\nu(\for_{n}((\gamma^{1},\ldots,\gamma^{d^{\prime\prime\prime}})))$.
Thus, for all $i{\in}\{1,\ldots,n\}$, by Lemma~\ref{lemma:7:about:beta:B:i:following:equivalent:conditions}, $f((\beta^{0},\ldots,\beta^{d^{\prime\prime}})){\models_{n}}x_{i}^{\gamma^{0}_{i}}$.
Moreover, by induction hypothesis, $f((\beta^{1},\ldots,\beta^{d^{\prime\prime}})){=}(\gamma^{1},\ldots,\gamma^{d^{\prime\prime\prime}})$.
Consequently, by Proposition~\ref{simple:proposition:about:morphisms}, $f((\beta^{0},\ldots,\beta^{d^{\prime\prime}})){=}(\gamma^{0},\ldots,\gamma^{d^{\prime\prime\prime}})$.
\end{description}
\vspace{0.93cm}
{\bf Proof of Lemma~\ref{lemma:9:about:beta:B:nu:tau:sigma:x:i}.}
Let $(\beta^{0},\ldots,\beta^{d^{\prime\prime}}){\in}\CHA_{d}^{k}$ and $i{\in}\{1,\ldots,n\}$.
\\
\\
$(i)\Rightarrow(ii)$ ---
Suppose $(\beta^{0},\ldots,\beta^{d^{\prime\prime}}){\models_{k}}\nu(\tau(x_{i}))$.
Let $(\alpha^{0},\ldots,\alpha^{d^{\prime}}){\in}\CHA_{d}^{k}$ be such that $(\alpha^{0},\ldots,\alpha^{d^{\prime}}){\models_{k}}\sigma(x_{i})$ and $(\beta^{0},\ldots,\beta^{d^{\prime\prime}}){\models_{k}}\nu(\for_{n}(f((\alpha^{0},\ldots,\alpha^{d^{\prime}}))))$.
Hence, by Lemma~\ref{lemma:8:about:beta:B:MOD:gamma:C:are:equi}, $f((\beta^{0},\ldots,\beta^{d^{\prime\prime}})){=}f((\alpha^{0},\ldots,\alpha^{d^{\prime}}))$.
Thus, $g((\beta^{0},\ldots,\beta^{d^{\prime\prime}})){=}g((\alpha^{0},\ldots,\alpha^{d^{\prime}}))$.
Since $(\alpha^{0},\ldots,\alpha^{d^{\prime}}){\models_{k}}\sigma(x_{i})$, therefore $(\beta^{0},\ldots,\beta^{d^{\prime\prime}}){\models_{k}}\sigma(x_{i})$.
\\
\\
$(ii)\Rightarrow(i)$ ---
Suppose $(\beta^{0},\ldots,\beta^{d^{\prime\prime}}){\models_{k}}\sigma(x_{i})$.
Consequently, $(\beta^{0},\ldots,\beta^{d^{\prime\prime}}){\models_{k}}\nu(\for_{n}(f((\beta^{0},\ldots,\beta^{d^{\prime\prime}}))))\rightarrow\nu(\tau(x_{i}))$.
Since by Lemma~\ref{lemma:8:about:beta:B:MOD:gamma:C:are:equi}, $(\beta^{0},\ldots,\beta^{d^{\prime\prime}}){\models_{k}}\nu(\for_{n}(f((\beta^{0},\ldots,\beta^{d^{\prime\prime}}))))$, therefore $(\beta^{0},\ldots,\beta^{d^{\prime\prime}}){\models_{k}}\nu(\tau(x_{i}))$.
\\
\\
\\
\\
{\bf Proof of Lemma~\ref{first:lemma:about:f:ast:alt1}.}
Let $(\beta^{0}){\in}\CHA_{=0}^{n}$.
We consider the following $2$ cases.
\\
\\
{\bf Case $(\beta^{0}){\in}\CHA_{=0}^{n}\setminus U$.}
Since $f^{\ast}_{0}$ is one-to-one, therefore let $(\alpha^{0}){\in}\CHA_{=0}^{k}$ be such that $(\alpha^{0}){\in}S$ and $f^{\ast}_{0}((\alpha^{0})){=}(\beta^{0})$.
Consequently, $f_{0}((\alpha^{0})){=}f^{\ast}_{0}((\alpha^{0}))$.
Since $f^{\ast}_{0}((\alpha^{0})){=}(\beta^{0})$, therefore $f_{0}((\alpha^{0})){=}(\beta^{0})$.
\\
\\
{\bf Case $(\beta^{0}){\not\in}\CHA_{=0}^{n}\setminus U$.}
Thus, $(\beta^{0}){\in}U$.
Consequently, let $(\alpha^{0}){\in}\CHA_{=0}^{k}$ be such that $(\alpha^{0}){=}h((\beta^{0}))$.
Hence, $f_{0}((\alpha^{0})){=}g((\alpha^{0}))$.
Since $g((\alpha^{0})){=}(\beta^{0})$, therefore $f_{0}((\alpha^{0})){=}(\beta^{0})$.
\\
\\
\\
\\
{\bf Proof of Lemma~\ref{second:lemma:about:f:ast:alt1}.}
Let $(\alpha^{0}),(\beta^{0}){\in}\CHA_{=0}^{k}$.
Suppose $f_{0}((\alpha^{0})){=}f_{0}((\beta^{0}))$.
We consider the following $3$ cases.
\\
\\
{\bf Case $\alpha^{0}{\in}S$ and $\beta^{0}{\in}S$.}
Hence, $f_{0}((\alpha^{0})){=}f^{\ast}_{0}((\alpha^{0}))$ and $f_{0}((\beta^{0})){=}f^{\ast}_{0}((\beta^{0}))$.
Since $f_{0}((\alpha^{0})){=}f_{0}((\beta^{0}))$, therefore $f^{\ast}_{0}((\alpha^{0})){=}f^{\ast}_{0}((\beta^{0}))$.
Since $f^{\ast}_{0}$ is one-to-one, therefore $\alpha^{0}{=}\beta^{0}$.
Consequently, $g((\alpha^{0})){=}g((\beta^{0}))$.
\\
\\
{\bf Case $\alpha^{0}{\in}S$ and $\beta^{0}{\not\in}S$.}
Thus, $f_{0}((\alpha^{0})){=}f^{\ast}_{0}((\alpha^{0}))$ and $f_{0}((\beta^{0})){=}g((\beta^{0}))$.
Since
\linebreak
$f_{0}((\alpha^{0})){=}f_{0}((\beta^{0}))$, therefore $f^{\ast}_{0}((\alpha^{0})){=}g((\beta^{0}))$.
Since $f^{\ast}_{0}((\alpha^{0})){\in}\CHA_{=0}^{n}\setminus U$ and $g((\beta^{0})){\in}U$, therefore $\CHA_{=0}^{n}\setminus U$ and $U$ do not make a partition of $\CHA_{=0}^{n}$: a contradiction.
\\
\\
{\bf Case $\alpha^{0}{\not\in}S$ and $\beta^{0}{\not\in}S$.}
Hence, $f_{0}((\alpha^{0})){=}g((\alpha^{0}))$ and $f_{0}((\beta^{0})){=}g((\beta^{0}))$.
Since
\linebreak
$f_{0}((\alpha^{0})){=}f_{0}((\beta^{0}))$, therefore $g((\alpha^{0})){=}g((\beta^{0}))$.
\\
\\
\\
\\
{\bf Proof of Lemma~\ref{first:lemma:about:f:ast:alt1:bis}.}
Let $(\delta^{0},\ldots,\delta^{d^{\prime}}){\in}\CHA_{=d^{\prime}}^{n}$.
Hence, $(\delta^{0},\ldots,\delta^{d^{\prime}}){\in}T((\delta^{1},\ldots,\delta^{d^{\prime}}))$.
Since $f_{d^{\prime}}^{(\delta^{1},\ldots,\delta^{d^{\prime}})}$ is surjective, therefore let $(\beta^{0},\ldots,\beta^{d^{\prime}}){\in}S((\delta^{1},\ldots,\delta^{d^{\prime}}))$ be such that $f_{d^{\prime}}^{(\delta^{1},\ldots,\delta^{d^{\prime}})}((\beta^{0},\ldots,\beta^{d^{\prime}})){=}(\delta^{0},\ldots,\delta^{d^{\prime}}))$.
Thus, $f_{d^{\prime}{-}1}((\beta^{1},\ldots,\beta^{d^{\prime}})){=}(\delta^{1},\ldots,\delta^{d^{\prime}})$.
Moreover, $f_{d^{\prime}}((\beta^{0},\ldots,\beta^{d^{\prime}})){=}f_{d^{\prime}}^{f_{d^{\prime}{-}1}((\beta^{1},\ldots,\beta^{d^{\prime}}))}((\beta^{0},\ldots,\beta^{d^{\prime}}))$.
\linebreak
Consequently, $f_{d^{\prime}}((\beta^{0},\ldots,\beta^{d^{\prime}})){=}f_{d^{\prime}}^{(\delta^{1},\ldots,\delta^{d^{\prime}})}((\beta^{0},\ldots,\beta^{d^{\prime}}))$.
Since $f_{d^{\prime}}^{(\delta^{1},\ldots,\delta^{d^{\prime}})}((\beta^{0},\ldots,\beta^{d^{\prime}})){=}(\delta^{0},\ldots,\delta^{d^{\prime}})$, therefore $f_{d^{\prime}}((\beta^{0},\ldots,\beta^{d^{\prime}})){=}(\delta^{0},\ldots,\delta^{d^{\prime}})$.
\\
\\
\\
\\
{\bf Proof of Lemma~\ref{second:lemma:about:f:ast:alt1:bis}.}
Let $(\alpha^{0},\ldots,\alpha^{d^{\prime}}),(\beta^{0},\ldots,\beta^{d^{\prime}}){\in}\CHA_{=d^{\prime}}^{k}$.
Suppose $f_{d^{\prime}}((\alpha^{0},\ldots,\alpha^{d^{\prime}})){=}f_{d^{\prime}}((\beta^{0},\ldots,\beta^{d^{\prime}}))$.
We consider the following $2$ cases.
\\
\\
{\bf Case $d^{\prime}{=}0$.}
Since $f_{d^{\prime}}((\alpha^{0},\ldots,\alpha^{d^{\prime}})){=}f_{d^{\prime}}((\beta^{0},\ldots,\beta^{d^{\prime}}))$, therefore by Lemma~\ref{second:lemma:about:f:ast:alt1},
\linebreak
$g((\alpha^{0},\ldots,\alpha^{d^{\prime}})){=}g((\beta^{0},\ldots,\beta^{d^{\prime}}))$.
\\
\\
{\bf Case $d^{\prime}{\geq}1$.}
Hence, $f_{d^{\prime}}((\alpha^{0},\ldots,\alpha^{d^{\prime}})){=}f_{d^{\prime}}^{f_{d^{\prime}{-}1}((\alpha^{1},\ldots,\alpha^{d^{\prime}}))}((\alpha^{0},\ldots,\alpha^{d^{\prime}}))$ and $f_{d^{\prime}}((\beta^{0},\ldots,\beta^{d^{\prime}})){=}f_{d^{\prime}}^{f_{d^{\prime}{-}1}((\beta^{1},\ldots,\beta^{d^{\prime}}))}((\beta^{0},\ldots,\beta^{d^{\prime}}))$.
Since $f_{d^{\prime}}((\alpha^{0},\ldots,\alpha^{d^{\prime}})){=}f_{d^{\prime}}((\beta^{0},\ldots,\beta^{d^{\prime}}))$, therefore $f_{d^{\prime}}^{f_{d^{\prime}{-}1}((\alpha^{1},\ldots,\alpha^{d^{\prime}}))}((\alpha^{0},\ldots,\alpha^{d^{\prime}})){=}f_{d^{\prime}}^{f_{d^{\prime}{-}1}((\beta^{1},\ldots,\beta^{d^{\prime}}))}((\beta^{0},\ldots,\beta^{d^{\prime}}))$.
\linebreak
Let $(\gamma^{1},\ldots,\gamma^{d^{\prime}}),(\delta^{1},\ldots,\delta^{d^{\prime}}){\in}\CHA_{=d^{\prime}{-}1}^{n}$ be such that $f_{d^{\prime}{-}1}((\alpha^{1},\ldots,\alpha^{d^{\prime}}))=(\gamma^{1},\ldots,\gamma^{d^{\prime}})$ and $f_{d^{\prime}{-}1}((\beta^{1},\ldots,\beta^{d^{\prime}}))=(\delta^{1},\ldots,\delta^{d^{\prime}})$.
Since $f_{d^{\prime}}^{f_{d^{\prime}{-}1}((\alpha^{1},\ldots,\alpha^{d^{\prime}}))}((\alpha^{0},\ldots,\alpha^{d^{\prime}})){=}f_{d^{\prime}}^{f_{d^{\prime}{-}1}((\beta^{1},\ldots,\beta^{d^{\prime}}))}((\beta^{0},\ldots,\beta^{d^{\prime}}))$, therefore $f_{d^{\prime}}^{(\gamma^{1},\ldots,\gamma^{d^{\prime}})}((\alpha^{0},\ldots,\alpha^{d^{\prime}})){=}f_{d^{\prime}}^{(\delta^{1},\ldots,\delta^{d^{\prime}})}((\beta^{0},\ldots,\beta^{d^{\prime}}))$.
Since $f_{d^{\prime}}^{(\gamma^{1},\ldots,\gamma^{d^{\prime}})}((\alpha^{0},\ldots,\alpha^{d^{\prime}})){\in}T((\gamma^{1},\ldots,\gamma^{d^{\prime}}))$ and
\linebreak
$f_{d^{\prime}}^{(\delta^{1},\ldots,\delta^{d^{\prime}})}((\alpha^{0},\ldots,\alpha^{d^{\prime}})){\in}T((\delta^{1},\ldots,\delta^{d^{\prime}}))$, therefore $(\gamma^{1},\ldots,\gamma^{d^{\prime}}){=}(\delta^{1},\ldots,\delta^{d^{\prime}})$.
\linebreak
Since $f_{d^{\prime}}^{(\gamma^{1},\ldots,\gamma^{d^{\prime}})}((\alpha^{0},\ldots,\alpha^{d^{\prime}})){=}f_{d^{\prime}}^{(\delta^{1},\ldots,\delta^{d^{\prime}})}((\beta^{0},\ldots,\beta^{d^{\prime}}))$, therefore $(\alpha^{0},\ldots,\alpha^{d^{\prime}}){\sim_{(\gamma^{1},\ldots,\gamma^{d^{\prime}})}}(\beta^{0},\ldots,\beta^{d^{\prime}})$ and $(\alpha^{0},\ldots,\alpha^{d^{\prime}}){\sim_{(\delta^{1},\ldots,\delta^{d^{\prime}})}}(\beta^{0},\ldots,\beta^{d^{\prime}})$.
Consequently,
\linebreak
$g((\alpha^{0},\ldots,\alpha^{d^{\prime}})){=}g((\beta^{0},\ldots,\beta^{d^{\prime}}))$.
\end{document}